\documentclass{amsart}

\usepackage{amsmath, amsthm, amssymb, bm, comment}
\usepackage{ascmac}

\usepackage{hyperref}
\usepackage{cleveref}
\usepackage[textsize=small, shadow]{todonotes}
\usepackage{cancel}

\numberwithin{equation}{section}

\newcommand{\bC}{\mathbb{C}}
\newcommand{\bR}{\mathbb{R}}

\newcommand{\U}{\mathrm{U}}

\newcommand{\sgn}{\mathrm{sgn}}

\def\la{\lambda}

\DeclareMathOperator*{\diag}{diag}

\DeclareMathOperator{\pr}{pr}
\DeclareMathOperator{\Tr}{Tr}

\DeclareMathOperator{\Wg}{Wg}

\theoremstyle{plain}
\newtheorem{thm}{Theorem}[section]

\newtheorem{lem}[thm]{Lemma}
\newtheorem{cor}[thm]{Corollary}
\theoremstyle{definition}

\newtheorem{example}{Example}[section]
\newtheorem{remark}{Remark}[section]
\theoremstyle{problem}

\title{Moments of Random Quantum Marginals via Weingarten Calculus}

\author{Sho Matsumoto}
\address{Graduate School of Science and Engineering, Kagoshima University}
\email{shom@sci.kagoshima-u.ac.jp}

\author{Colin McSwiggen}
\address{Courant Institute of Mathematical Sciences, New York University}
\email{csm482@nyu.edu}

\begin{document}
%>>>>>>>>>>>>>>>>>>>>>>>>>>>>>>>>>>>>>>>>>

\maketitle

\begin{abstract}
The randomized quantum marginal problem asks about the joint distribution of the partial traces (``marginals'') of a uniform random Hermitian operator with fixed spectrum acting on a space of tensors.  We introduce a new approach to this problem based on studying the mixed moments of the entries of the marginals.  For randomized quantum marginal problems that describe systems of distinguishable particles, bosons, or fermions, we prove formulae for these mixed moments, which determine the joint distribution of the marginals completely.  Our main tool is Weingarten calculus, which provides a method for computing integrals of polynomial functions with respect to Haar measure on the unitary group.  As an application, in the case of two distinguishable particles, we prove some results on the asymptotic behavior of the marginals as the dimension of one or both Hilbert spaces goes to infinity.
\end{abstract}

\section{Introduction}

In classical mechanics, a system is just the sum of its parts.  For a system comprised of two classical particles, if I know the state (i.e., position and momentum) of each individual particle, then I know everything about the system.  In quantum mechanics, this is not the case.  Due to the phenomenon of entanglement, the state of a quantum mechanical system contains more information than the combined states of its subsystems do.

Mathematically, we can express this fact in terms of operators on Hilbert spaces.  The formalism of quantum mechanics associates, to a given physical system $S$, a Hilbert space $\mathcal{H}(S)$. (In this paper we will deal only with finite-dimensional Hilbert spaces, although the problems we study can be posed in greater generality.) The tensor product postulate of quantum mechanics states that if $S$ is a composite system formed from disjoint subsystems $S_1$ and $S_2$, then $\mathcal{H}(S) \cong \mathcal{H}(S_1) \otimes \mathcal{H}(S_2)$.

States of the system $S$ can be identified with their \textit{density matrices}, which are positive semidefinite Hermitian operators on $\mathcal{H}(S)$ with trace 1.  If $e_1, \hdots, e_m$ is a basis of $\mathcal{H}(S_1)$ and $f_1, \hdots, f_n$ is a basis of $\mathcal{H}(S_2)$, then the tensors $e_i \otimes f_j$ form a basis of $\mathcal{H}(S) \cong \mathcal{H}(S_1) \otimes \mathcal{H}(S_2)$.  Accordingly, we write the matrix entries of a density matrix $M$ on $\mathcal{H}(S)$ using double indices $11, 12, \hdots, mn$, ordered lexicographically:
$$M = (M_{ij,kl})_{ij,kl=11}^{mn}.$$
The {\it partial traces} of such a matrix are matrices $\pi_1(M)$ and $\pi_2(M)$, which are density matrices acting on $\mathcal{H}(S_1)$ and $\mathcal{H}(S_2)$ respectively, with entries
\begin{align*}
\pi_1(M)_{i,k} &= \sum_{j=1}^n M_{ij, kj}, \qquad i,k=1,\hdots,m, \\
\pi_2(M)_{j,l} &= \sum_{i=1}^m M_{ij, il}, \qquad j,l = 1, \hdots, n.
\end{align*}
These partial traces are also called the \textit{marginals} of $M$, and they have a physical meaning.  If the system $S$ is in the state $M$ and we observe only the individual subsystem $S_1$, then $S_1$ will appear to be in the state $\pi_1(M)$.  Likewise, if we observe only the individual subsystem $S_2$, then $S_2$ will appear to be in the state $\pi_2(M)$.  However, an easy dimension count reveals that in general it is not possible to reconstruct the matrix $M$ given only $\pi_1(M)$ and $\pi_2(M)$.  This is the mathematical expression of the idea that in quantum mechanics, a system is more than the sum of its parts.

Even though we cannot hope to recover $M$ from $\pi_1(M)$ and $\pi_2(M)$, we can ask a more basic question: given a density matrix $A$ on $\mathcal{H}(S_1)$ and another density matrix $B$ on $\mathcal{H}(S_2)$, is there a density matrix $M$ on $\mathcal{H}(S)$ with $\pi_1(M) = A$ and $\pi_2(M) = B$?  Typically, we make the further stipulation that $M$ must have given eigenvalues $\lambda = (\lambda_{11}, \hdots, \lambda_{mn}) \in \bR^{mn}$.  This question is a fundamental problem in quantum mechanics known as the \textit{quantum marginal problem} (or as the $N$-\textit{representability problem} in theoretical chemistry).  It asks whether the states $A$ and $B$ are compatible, in the sense that $S_1$ could be in state $A$ and $S_2$ could be in state $B$ while the composite system $S$ is some state $M$ with the required spectrum $\lambda$.

The quantum marginal problem was first posed at least as early as the 1960's and has been an active topic of research in the decades since \cite{CY, Ruskai, Stillinger}.  In the case where $\mathcal{H}(S_1)$ and $\mathcal{H}(S_2)$ are both finite dimensional, the problem was solved in 2004 by Klyachko \cite{Klyachko-marginals}, who proved that the set of compatible pairs $(A,B)$ of density matrices on $\mathcal{H}(S_1)$ and $\mathcal{H}(S_2)$ is determined by a certain system of linear inequalities on the eigenvalues of $A$ and $B$ and on $\lambda$.\footnote{However, despite these explicit compatibility criteria, actually checking whether they hold is computationally difficult on both classical and quantum computers \cite{Liu,LCV}.}  See \cite{Knut-lectures} for a review of Klyachko's solution in the finite dimensional case, as well as \cite{Schilling-review, Tyc-Vlach} for more recent general reviews of quantum marginal problems from a physics perspective.

The setup where $\mathcal{H}(S) \cong \mathcal{H}(S_1) \otimes \mathcal{H}(S_2)$ is appropriate for describing a system of two distinguishable particles, where each particle corresponds to one of the subsystems $S_1$ or $S_2$.  However, many physical systems of interest consist of indistinguishable bosons (whose states are constrained to be invariant under permutations of the particles) or indistinguishable fermions (whose states change sign when any two of the particles are exchanged).  That is, for a system of $k$ indistinguishable bosons or fermions, states correspond to density matrices on $\mathrm{Sym}^k \bC^n$ or $\wedge^k \bC^n$ respectively, where $\bC^n$ represents the Hilbert space associated to any given single-particle subsystem.  Identifying $(\bC^n)^{\otimes k} \cong \bC^{n^k}$, we can regard such density matrices as $n^k$-by-$n^k$ matrices that satisfy a symmetry or antisymmetry constraint; single-particle marginals are obtained by taking the partial trace over all but one leg of the $k$-fold tensor product.  The quantum marginal problem in this setting asks: given $\lambda \in \bR^{n^k}$ and a density matrix $A$ on $\bC^n$, does there exist an $n^k$-by-$n^k$ density matrix $M$ with eigenvalues $\lambda$, {\it satisfying the (anti)symmetry constraint}, with a single-particle marginal equal to $A$?

The above problems are qualitative in character.  They all ask a yes--no question: does a state of the composite system with the given spectrum and marginals exist?  In this paper, we study random matrix ensembles that can be regarded as \textit{quantitative} versions of these quantum marginal problems.  Concretely, we ask: given a uniform random Hermitian matrix\footnote{In mathematical treatments of quantum marginal problems, it is standard to drop the requirement that the matrices in question be density matrices (i.e.~positive semidefinite Hermitian matrices with trace 1) and instead study the partial traces of Hermitian matrices with arbitrary fixed spectra.  From here on we will follow this convention.} acting on $\bC^m \otimes \bC^n$ (or, in the case of bosons or fermions, on $\mathrm{Sym}^k \bC^n$ or $\wedge^k \bC^n$) with fixed eigenvalues, what is the probability distribution of its partial traces?  This is a refinement of the original quantum marginal problem, in the sense that the original problem asks only about the support of this distribution, whereas now we would like to describe the distribution completely.

To illustrate this randomized version of the problem in more detail, we consider the example of two distinguishable particles.  Here we take integers $m, n \ge 2$ representing the dimensions of the Hilbert spaces of the single-particle subsystems, and we fix a vector of eigenvalues $\lambda \in \mathbb{R}^{mn}$.  Let $U$ be a random $mn \times mn$ unitary matrix distributed according to the Haar probability measure\footnote{The \textit{Haar probability measure} on any compact Lie group is the unique probability measure that is invariant under both the left and right actions of the group on itself.  A random group element distributed according to the Haar probability measure can be regarded as ``uniformly distributed'' on the group.} on $\U(mn)$, and set
\begin{equation} \label{eqn:H-intro}
H=(H_{ij, kl})_{ij,kl=11}^{mn} = U \diag (\lambda) U^\dagger.
\end{equation}
Then $H$ is a uniform random $mn \times mn$ Hermitian matrix with eigenvalues $\lambda$.  The randomized quantum marginal problem in this case asks for a description of the joint distribution of the marginals $(\pi_1(H), \pi_2(H))$.

Randomized quantum marginal problems have been studied in the literature using two different approaches.  In a seminal paper, Christandl, Doran, Kousidis and Walter \cite{CDKW} used techniques from symplectic geometry to give an algorithm for computing the spectral distributions of marginals for distinguishable particles, bosons, and fermions.  In later work, Collins and the second author \cite{CM} used Fourier analysis techniques to derive integral formulae for these same distributions and to prove various properties of their densities.  While each of these approaches offers a solution to these problems in a certain sense, the existing results are not easily amenable to every type of analysis one would like to perform, and many unanswered questions remain.

This paper introduces a third approach, based on computing the mixed moments of the entries of the marginals.  A major advantage of this approach is that, at least in the case of distinguishable particles, the theorems are much more amenable to asymptotic analysis, providing a way to study the limiting behavior of the marginals as the dimensions of the Hilbert spaces go to infinity.  Our method exploits the fact that any compactly supported probability measure on $\bR^n$ is completely characterized by its \textit{mixed moments}, that is, expectations of monomials
$$\mathbb{E} \bigg[ \prod_{j=1}^p x_{i_j} \bigg ], \qquad (i_1, \hdots, i_p) \in \{1, \hdots, n\}^{\times p},$$
for all $p \ge 1$, where $x_1, \hdots, x_n$ are the standard coordinates on $\bR^n$.  For probability measures on the line, this fact was first observed by Hausdorff \cite{Hausdorff-moments}.  The distribution of the random matrix $H$ in (\ref{eqn:H-intro}) is compactly supported due to the compactness of the unitary group; therefore the distribution of $(\pi_1(H), \pi_2(H))$ is also compactly supported, since the partial traces are linear maps.  Accordingly, \textit{all} information about the joint distribution of the marginals is contained in the mixed moments
\begin{equation} \label{eqn:mm-intro}
\mathbb{E} \left[ 
\prod_{\alpha=1}^p \pi_1(H)_{i_\alpha, k_\alpha} \cdot 
\prod_{\beta=1}^q \pi_2(H)_{j_\beta, l_\beta}
\right]
\end{equation}
for all $p,q \ge 0$, $i_1,\dots,i_p,k_1,\dots,k_p \in \{1,2,\dots,m\}$, 
and $j_1,\dots,j_q, l_1,\dots,l_q \in \{1,2,\dots,n\}$. (Note that the space of $n \times n$ complex Hermitian matrices is an $n^2$-dimensional {\it real} vector space, and (\ref{eqn:mm-intro}) indeed captures the expectations of all monomials in the coordinates.)  A formula for computing the mixed moments (\ref{eqn:mm-intro}) can thus be regarded as a third type of solution to the randomized quantum marginal problem, complementary to those given in \cite{CDKW} and \cite{CM}.

The main contributions of this paper are formulae of this type.  Specifically, we derive a formula for (\ref{eqn:mm-intro}), as well as analogous formulae for mixed moments of single-particle marginals for systems of $k$ distinguishable particles, $k$ indistinguishable bosons, or $k$ indistinguishable fermions.  The proofs of these formulae rely on {\it Weingarten calculus}, a set of techniques for computing integrals of polynomial functions with respect to invariant measures on compact groups and symmetric spaces \cite{Collins03, CMN-review}.  As an application, in the case of two distinguishable particles, we prove some results on the asymptotic behavior of the marginals as the dimension of one or both Hilbert spaces goes to infinity.

\begin{remark}
Beyond quantum information theory, randomized quantum marginal problems are a natural object of study from at least three different points of view:
\begin{itemize}
\item In \textbf{quantum statistical mechanics}, a significant body of research has studied the properties of ``typical'' quantum states, that is, random states that are uniformly distributed conditional on certain information \cite{BHKRV, LloydPagels, SZ1, SZ2, SZ3}.
\item In \textbf{algebraic combinatorics}, the randomized quantum marginal problem for two distinguishable particles with $n$-dimensional Hilbert spaces is a semiclassical approximation for the tensor product multiplicities of irreducible representations of the symmetric group $S_n$; see \cite[\textsection5]{Klyachko-marginals}, \cite[\textsection7]{CDKW}, \cite[\textsection4]{CM}.  These multiplicities are known as \textit{Kronecker coefficients}.  Aside from their intrinsic combinatorial and representation-theoretic interest, they have also been a central object of study in the Geometric Complexity Theory program \cite{BI,IMW,IP}.
\item Randomized quantum marginal problems are examples of random matrix ensembles derived from projections of invariant measures on coadjoint orbits of compact Lie groups, a family of models that has recently attracted a great deal of interest in \textbf{random matrix theory}.  Other ensembles of this type include the randomized Horn's problem \cite{Z, CZ1, CZ2, CMZ, CMZ2, McS-splines, BGH, ForresterZhang, ZKF}, the randomized Schur's problem  \cite{CMZ2, CZ-SchurKostka, BGH}, and the orbital corners process \cite{Barysh, CC-corners, Z2}.  This class of random matrix ensembles has recently been studied at a high level of generality in \cite{CM}.
\end{itemize}
\end{remark}

\subsection{Organization of the paper}
In \S\ref{sec:distinguishable}, we treat the randomized quantum marginal problem for distinguishable particles.  After reviewing the necessary background on Weingarten calculus, we derive formulae for the mixed moments of entries of single-particle marginals for systems of two distinguishable particles (Theorem \ref{thm:projection_distribution}) and then for $k$ distinguishable particles (Theorem \ref{thm:projection_multi_1}).

In \S\ref{sec:bos-fer}, we consider systems of $k$ indistinguishable bosons or fermions.  Here again we prove formulae for the mixed moments of entries of single-particle marginals (Theorems \ref{thm:bos-moments} and \ref{thm:fer-moments}).  We conclude with some observations that explain why the formulae for bosons and fermions appear to be less analytically tractable than those obtained for distinguishable particles.

In \S\ref{sec:asymptotics}, we return to the case of two distinguishable particles and study the asymptotic behavior of the marginals in two regimes: the regime where the dimension of one Hilbert space goes to infinity while the dimension of the other remains fixed, and the regime where the dimensions of both Hilbert spaces go to infinity while their ratio remains fixed.  We prove formulae for the leading-order contribution to each mixed moment in both of these regimes (Theorems \ref{thm:asymp-m} and \ref{thm:asymp-n}), and in the first regime we also prove a law of large numbers for the limiting distribution of the finite-dimensional marginal (Theorem \ref{thm:fd-marg-limit}).

\section{Distinguishable particles} \label{sec:distinguishable}

\subsection{The marginal problem for two distinguishable particles} \label{subsec:our_q}

We first study the mixed moments of the single-particle marginals of a quantum mechanical system consisting of two distinguishable particles.  Concretely, we address the following problem.

Let $m,n \ge 2$ be integers, which we take to represent the dimensions of the Hilbert spaces for each of two particles.
Fix $\lambda =(\lambda_{ij})_{ij=11}^{mn} \in \mathbb{R}^{mn}$, where $ij$ and $kl$ stand for double indices with $i,k \in \{1,2,\dots,m\}$ and $j,l \in \{1,2,\dots,n\}$.
Let $U$ be a Haar-distributed unitary matrix of size $mn$.
We consider a random Hermitian matrix given by 
\begin{equation}
H=(H_{ij, kl}) = U \diag (\lambda) U^\dagger.
\end{equation}
Define the $m \times m$ Hermitian matrix $\pi_1(H) =(\pi_1(H)_{i,k})$ and 
$n \times n$ Hermitian matrix $\pi_2(H)=(\pi_2(H)_{j,l})$ by
\begin{equation} \label{eq:proj_pi}
\pi_1(H)_{i,k}= \sum_{j=1}^n H_{ij, kj}, \qquad \pi_2(H)_{j,l}= \sum_{i=1}^m H_{ij, il}.
\end{equation}
Our aim is to provide a formula for the mixed moments 
\begin{equation}
\mathbb{E} \left[ 
\prod_{\alpha=1}^p \pi_1(H)_{i_\alpha, k_\alpha} \cdot 
\prod_{\beta=1}^q \pi_2(H)_{j_\beta, l_\beta}
\right]
\end{equation}
for $i_1,\dots,i_p,k_1,\dots,k_p \in \{1,2,\dots,m\}$ and 
$j_1,\dots,j_q, l_1,\dots,l_q \in \{1,2,\dots,n\}$.

\subsection{Notation for symmetric groups}
\label{subsec:sym_gp}

Let $d$ be a positive integer and let $\mathfrak{S}_d$ be the symmetric group on $d$ elements.
For a permutation $\sigma \in \mathfrak{S}_d$ and 
two sequences of $d$ positive integers
$\bm{a}=(a_1,\dots,a_d)$, $\bm{b}=(b_1,\dots,b_d)$,
we define
\begin{equation} \label{def:delta_function}
\delta_{\sigma}(\bm{a},\bm{b})= \prod_{r=1}^d \delta(a_{\sigma(r)}, b_r)
=\begin{cases}
1 & \text{if $a_{\sigma(r)}=b_r$ for all $r=1,2,\dots,d$} \\
0 & \text{otherwise}.
\end{cases}
\end{equation}
Here $\delta(i,j)$ stands for Kronecker's delta\footnote{We avoid the use of the standard notation $\delta_{ij}$, as it would force us to write unsightly symbols like $\delta_{a_{\alpha_2},b_{\alpha_1}}$.
}.
In some cases, it is convenient to apply the cycle decomposition of $\sigma$. 
For a cycle $c=(\alpha_1 \ \alpha_2 \ \dots \  \alpha_h)$ of $\sigma$, 
we define 
\begin{equation} 
\delta_{c}(\bm{a},\bm{b})
=\delta(a_{\alpha_2},b_{\alpha_1}) \delta(a_{\alpha_3},b_{\alpha_2})
\cdots \delta( a_{\alpha_h},b_{\alpha_{h-1}}) \delta(a_{\alpha_1},b_{\alpha_h})
\end{equation} 
if $h \ge 2$, or $\delta_{c}(\bm{a},\bm{b})=\delta(a_{\alpha_1},b_{\alpha_1})$ if $h=1$.
Then it is immediate that 
\begin{equation} \label{eq:delta-cycle-decomposition}
\delta_\sigma(\bm{a},\bm{b})= \prod_{c} \delta_c(\bm{a},\bm{b}),
\end{equation}
where the product runs over all cycles $c$ in $\sigma$.

For example, if $d=10$ and 
\begin{equation} \label{eq:example_permutation}
\sigma=\begin{pmatrix} 
1 & 2 & 3 & 4 & 5 & 6 & 7 & 8 & 9 & 10 \\
3 & 7 & 1 & 9 & 5 & 4 & 6 & 10 & 2 & 8 
\end{pmatrix}
=(1\ 3)(2 \ 7 \ 6 \ 4 \ 9)(5)(8 \ 10),
\end{equation}
we see that
\begin{align*}
\delta_{\sigma} (\bm{a}, \bm{b})
&= 
\delta(a_{3}, b_1) \delta(a_7, b_2) \delta(a_1, b_3) \delta(a_9, b_4) 
\delta(a_5, b_5) \\
& \cdot \delta(a_4, b_6) \delta(a_6, b_7) \delta(a_{10},b_8) 
\delta(a_2, b_9) \delta(a_8,b_{10}) \\
&= \{ \delta(a_{3}, b_1)\delta(a_1, b_3) \} \cdot 
 \{ \delta(a_7, b_2) \delta(a_6, b_7)\delta(a_4, b_6)\delta(a_9, b_4) 
 \delta(a_2, b_9)  \} \\
 & \cdot \{ \delta(a_5, b_5) \} \cdot \{  \delta(a_{10},b_8)  \delta(a_8,b_{10})\}
\\
&= \delta_{(1 \ 3)} (\bm{a},\bm{b}) \cdot
\delta_{(2 \ 7 \ 6 \ 4 \ 9)} (\bm{a},\bm{b}) \cdot 
\delta_{(5)} (\bm{a},\bm{b}) \cdot \delta_{(8 \ 10)} (\bm{a},\bm{b}).
\end{align*}

Let $p,q$ be nonnegative integers with $p+q \ge 1$. 
We define projections
\begin{align*}
	\pr_1:& \ \mathfrak{S}_{p+q} \to \mathfrak{S}_p, \\
	\pr_2:& \ \mathfrak{S}_{p+q} \to \mathfrak{S}_q
\end{align*}
as follows.  For each $\sigma \in \mathfrak{S}_{p+q}$, we let $\pr_1(\sigma) \in \mathfrak{S}_p$ be the permutation on $p$ elements obtained by decomposing $\sigma$ into a product of cycles and erasing letters $p+1, p+2,\dots,p+q$
from each cycle.  For example, for $p=q=5$ and $\sigma$ given in \eqref{eq:example_permutation}, 
we obtain 
$\pr_1(\sigma)= (1\ 3)(2 \ {\cancel{7}} \ \cancel{6} \ 4 \ \cancel{9})(5)(\cancel{8} \ \cancel{10}) =(1 \ 3)(2 \ 4)(5) \in \mathfrak{S}_{5}$.

Similarly, we let $\pr_2(\sigma) \in \mathfrak{S}_q$ be the permutation on $q$ elements obtained by decomposing $\sigma$ into a product of cycles, erasing the letters $1, 2,\dots,p$, and replacing the remaining letters $p+1,p+2,\dots, p+q$ by $1,2,\dots, q$, respectively.  For example, for again the same $p, q$ and $\sigma$, the elimination of letters $1,2,3,4,5$ provides 
$(\cancel{1}\ \cancel{3})(\cancel{2} \ 7 \ 6 \ \cancel{4} \ 9)(\cancel{5})(8 \ 10)=(7 \ 6 \ 9)(8 \ 10)$. 
When we replace letters $6,7,8,9,10$ by $1,2,3,4,5$ respectively, we obtain  
$\pr_2(\sigma)= (2 \ 1 \ 4)(3 \ 5) \in \mathfrak{S}_{5}$.

Next we define
quantities $\kappa_1(\sigma)$ and $\kappa_2(\sigma)$ 
for $\sigma \in \mathfrak{S}_{p+q}$ in the following way. 
\begin{itemize}
\item $\kappa_1(\sigma)$ is the number of cycles $c$ in $\sigma$ satisfying $c \cap \{1,2,\dots, p\} = \emptyset$,
the latter of which means that the orbit of the cycle $c$ is included in the set $\{p+1,p+2,\dots, p+q\}$.
\item $\kappa_2(\sigma)$ is the number of cycles $c$ in $\sigma$ satisfying $c \cap \{p+1,p+2,\dots, p+q\} = \emptyset$.
\end{itemize}
For example, 
let $p=q=5$, $\sigma=(1\ 3)(2 \ 7 \ 6)(4)(5 \ 9)(8 \ 10)$. Then:
\begin{itemize}
\item Two cycles $(1 \ 3)$ and $(4)$ contribute to $\kappa_2(\sigma)$,
because orbits $\{1,3\}$ and $\{4\}$ of $\sigma$ are included in $\{1,2,3,4,5\}$;
\item  One cycle $(8 \ 10)$ contributes to $\kappa_1(\sigma)$, 
because the orbit $\{8,10\}$ is included in $\{6,7,8,9,10\}$; 
\end{itemize}
We therefore have $\kappa_2(\sigma)=2$, $\kappa_1(\sigma)=1$.

For any permutation $\sigma \in \mathfrak{S}_d$, we write $\kappa(\sigma)$, with no subscript, for the total number of cycles in the cycle decomposition of $\sigma$.

\subsection{Background on Weingarten calculus}
\label{subsec:Wg}

Our main tool in what follows will be {\it Weingarten calculus}, a set of techniques for computing integrals of polynomial functions with respect to invariant measures on compact groups and symmetric spaces.  Weingarten calculus was introduced in \cite{Collins03} and has since been developed by numerous authors.  Here we merely recall some basic definitions and facts from Weingarten calculus, and refer the reader to \cite{CMN-review} for a more comprehensive introduction and survey of the literature.

Recall that for $d$ a positive integer, the irreducible representations of $\mathfrak{S}_d$ are indexed by partitions $\lambda \vdash d$, i.e.~weakly decreasing sequences of nonnegative integers $(\lambda_1, \hdots, \lambda_l)$ such that $\lambda_1 + \cdots + \lambda_l = d$.\footnote{See e.g.~\cite{Ful97} for background on the representation theory of symmetric groups.} For each such partition $\lambda = (\lambda_1, \hdots, \lambda_l)$, we write $\ell(\lambda) = l$ for the length of $\lambda$ and $\chi^\lambda$ for the corresponding irreducible character of $\mathfrak{S}_d$.  For a positive integer $N$, set
$$C_\lambda(N) = \prod_{i = 1}^{\ell(\lambda)} \prod_{j=1}^{\lambda_i} (N + j - i).$$
The {\it unitary Weingarten function} $\Wg_{N,d} : \mathfrak{S}_d \to \bC$ is defined by
\begin{equation} \label{eqn:Wg-def}
\Wg_{N,d}(\sigma) = \frac{1}{d!} \sum_{\substack{\lambda \vdash d \\ C_\lambda(N) \ne 0}} \frac{\chi^\lambda(\mathrm{id}_d) \chi^\lambda(\sigma)}{C_\lambda(N)}, \qquad \sigma \in \mathfrak{S}_d,
\end{equation}
where $\mathrm{id}_d \in \mathfrak{S}_d$ is the identity permutation.  Where the value of $d$ can be understood from context, we will usually suppress the explicit dependence on $d$ in the notation and simply write $\Wg_{N}$ for $\Wg_{N,d}$.

Some explicit values for $\Wg_{N,d}$ are as follows.
\begin{align}
\Wg_{N,1}((1)) &= \frac{1}{N}, \label{Wg_U_1} \\
\Wg_{N,2}((1)(2)) &= \frac{1}{(N+1)(N-1)}, 
\label{Wg_U_11}\\
\Wg_{N,2}((1 \ 2)) &= \frac{-1}{N(N+1)(N-1)}. \label{Wg_U_2}
\end{align}

The asymptotics of $\Wg_{N,d}$ as $N$ grows large are well understood.  Notably, the following corollary of \cite[Theorem 2.2]{Collins03} describes the leading-order behavior.

\begin{lem} \label{lem:Wg-asymp}
For a fixed positive integer $d$ and $\sigma \in \mathfrak{S}_d$, as $N \to \infty$,
\begin{equation} \label{eqn:Wg-asymp}
\Wg_{N,d}(\sigma) = 
\begin{cases}
N^{-d} \cdot (1+O(N^{-1})) & \text{if $\sigma=\mathrm{id}_d$}, \\
O(N^{-d-1}) & \text{otherwise}.
\end{cases}
\end{equation}
\end{lem}

Weingarten calculus gives a method for writing the expectations of polynomial functions of the entries of a Haar-distributed random unitary matrix in terms of values of $\Wg_N$.  In particular, below we will make extensive use of the following result, which is a special case of \cite[Theorem 3.1]{CMS}.

\begin{lem}
\label{lem:CMS}
Let $N,d$ be positive integers.
Fix a real Hermitian matrix $A$ of size $N$.
Let $U$ be a Haar-distributed matrix from the unitary group $\mathrm{U}(N)$. For a random matrix $W$ given by $W=(W_{ik})_{i,k=1}^N = U A U^{\dagger}$ and for two sequences 
\[
\bm{i}=(i_1,\dots,i_d), \quad \bm{k}=(k_1,\dots, k_d) \quad \in \{1,2,\dots,N\}^{\times d}, 
\]
we have 
\[
\mathbb{E}[W_{i_1,k_1} W_{i_2, k_2} \cdots W_{i_d, k_d}] =
\sum_{\sigma, \tau \in \mathfrak{S}_d} \delta_\sigma(\bm{i},\bm{k}) 
\Wg_{N}(\sigma^{-1}\tau) \Tr_\tau(A).
\]
Here
\begin{itemize}
\item  
$\delta_{\sigma}(\bm{i},\bm{k})$ is defined in \eqref{def:delta_function};
\item $\Wg_{N}$ is the unitary Weingarten function defined in (\ref{eqn:Wg-def}); 
\item If the cycle type of $\tau \in \mathfrak{S}_d$ is $(\mu_1,\mu_2,\dots)$, set 
\begin{equation}
\Tr_\tau(A) := \prod_{j \ge 1} \Tr(A^{\mu_j}).
\end{equation}
\end{itemize}
\end{lem}

\subsection{The moment formula for two distinguishable particles}

We are now ready to state our first main result.  Let $\lambda, H, \pi_1(H), \pi_2(H)$ be as in \S \ref{subsec:our_q}.

\begin{thm} \label{thm:projection_distribution}
Let $p, q$ be nonnegative integers. For sequences of indices
\begin{align*}
\bm{i}=(i_1,i_2,\dots,i_p), \ \bm{k}=(k_1,k_2,\dots,k_p)  \quad \in \{1,2,\dots,m\}^{\times p} &, \\ 
\bm{j}=(j_1, j_2, \dots, j_q), \ \bm{l}=(l_1,l_2,\dots,l_q)  \quad \in \{1,2,\dots,n\}^{\times q}, &
\end{align*}
we have the formula
\begin{multline}
\mathbb{E} \left[ 
\prod_{\alpha=1}^p \pi_1(H)_{i_\alpha, k_\alpha} \cdot 
\prod_{\beta=1}^q \pi_2(H)_{j_\beta, l_\beta}
\right] \\
= \sum_{\sigma ,\tau \in \mathfrak{S}_{p+q}} 
\delta_{\pr_1 (\sigma)} (\bm{i},\bm{k}) \, \delta_{\pr_2 (\sigma)} (\bm{j},\bm{l}) \,
n^{\kappa_2(\sigma)} m^{\kappa_1(\sigma)} \Wg_{mn}(\sigma^{-1}\tau) \Tr_\tau(\lambda),
\end{multline}
with the following notation:
\begin{itemize}
\item $\delta$ is defined in \eqref{def:delta_function};
\item $\pr_1$, $\pr_2$, $\kappa_1$, $\kappa_2$ are defined in \S \ref{subsec:sym_gp};
\item $\Wg_{mn}$ and $\Tr_\tau$ are defined in \S\ref{subsec:Wg},
but $\Tr_\tau(\lambda) :=\Tr_\tau(\diag(\lambda))$.
\end{itemize}
\end{thm}

\subsection{Examples}

If we set $q=0$ or $p=0$ in Theorem \ref{thm:projection_distribution}, 
then we find, respectively:
\begin{align}
\mathbb{E} \left[ 
\prod_{\alpha=1}^p \pi_1 (H)_{i_\alpha, k_\alpha} \right] 
&= \sum_{\sigma ,\tau \in \mathfrak{S}_{p}} 
\delta_{\sigma} (\bm{i},\bm{k}) \,
n^{\kappa(\sigma)} \Wg_{mn}(\sigma^{-1}\tau) \Tr_\tau(\lambda), \\
\mathbb{E} \left[ 
\prod_{\beta=1}^q \pi_2 (H)_{j_\beta, l_\beta} \right] 
&= \sum_{\sigma ,\tau \in \mathfrak{S}_{q}} 
\delta_{\sigma} (\bm{j},\bm{l}) \,
m^{\kappa(\sigma)} \Wg_{mn}(\sigma^{-1}\tau) \Tr_\tau(\lambda),
\end{align}
where $\kappa(\sigma)$ indicates the total number of cycles in $\sigma$, as defined in  \S \ref{subsec:sym_gp}.  In particular, using \eqref{Wg_U_1}, 
\begin{equation}
\mathbb{E} \left[ \pi_1 (H)_{ik} \right]
= \delta(i,k) \, m^{-1} \Tr(\lambda),
\qquad 
\mathbb{E} \left[ \pi_2 (H)_{jl} \right]
= \delta(j,l) \, n^{-1} \Tr(\lambda),
\end{equation}
where
$\Tr(\lambda)= \lambda_{11}+\lambda_{12}+ \cdots+ \lambda_{mn}$.

Next, if we set $p=q=1$ in Theorem \ref{thm:projection_distribution}, 
then 
\begin{equation} \label{eq:p=q=1}
\mathbb{E} \left[ 
\pi_1 (H)_{i, k} \, \pi_2(H)_{j,l}
\right] 
= \delta(i,k) \, \delta( j,l) \, (mn)^{-1} \left(\Tr(\lambda)\right)^2.
\end{equation}
To see this, note that
for each $\sigma \in \mathfrak{S}_2$, we have $\pr_{1}(\sigma)=\pr_2(\sigma)= (1) \in \mathfrak{S}_1$,
so that
$$\delta_{\pr_1 (\sigma)} (i,k)  \delta_{\pr_2 (\sigma)} (j,l)
=\delta(i,k) \delta(j,l).$$
Theorem \ref{thm:projection_distribution} implies 
\begin{align*}
& \mathbb{E} \left[ 
\pi_1 (H)_{i, k}  \pi_2(H)_{j,l}\right] \\
&=  \delta(i,k) \delta(j,l)
\Big\{ 
 \underbrace{ n m \Wg_{mn} ( (1)(2) ) \Tr_{(1) (2)}(\lambda) }_{\sigma=\tau=(1)(2)}  
+\underbrace{\Wg_{mn} ( (1 \ 2) ) \Tr_{(1) (2)}(\lambda)}_{\sigma=(1 \ 2),\ \tau=(1)(2) }   
  \\
& \qquad + \underbrace{ n m \Wg_{mn} ( (1 \ 2) )  \Tr_{(1 \ 2)}(\lambda)}_{\sigma=(1)(2), \ 
\tau=(1 \ 2)}
+\underbrace{\Wg_{mn} ( (1)(2) )  \Tr_{(1 \ 2)}(\lambda)}_{\sigma=\tau=(1 \ 2)}  
\Big\}.
\end{align*}
Using \eqref{Wg_U_11} and \eqref{Wg_U_2}
we have
\begin{align*}
mn \Wg_{mn} ( (1)(2) ) + \Wg_{mn} ( (1 \ 2) )  
&= \frac{1}{mn},  \\
mn \Wg_{mn} ( (1 \ 2) ) + \Wg_{mn} ( (1 )(2) )  
&= 0. 
\end{align*}
Thus, we have proved \eqref{eq:p=q=1}.

\subsection{Proof of Theorem \ref{thm:projection_distribution}}

We start by introducing some notation to simplify handling the sequences of indices
\begin{align*}
\bm{i}=(i_1,i_2,\dots,i_p), \ \bm{k}=(k_1,k_2,\dots,k_p)  \quad \in \{1,2,\dots,m\}^{\times p} &, \\ 
\bm{j}=(j_1, j_2, \dots, j_q), \ \bm{l}=(l_1,l_2,\dots,l_q)  \quad \in \{1,2,\dots,n\}^{\times q}. &
\end{align*}
We write $\bm{i} \cup \bm{s}$ for the sequence $(i_1,i_2,\dots,i_p,s_1,s_2,\dots, s_q)$ of length $p+q$, and we define $\bm{k}\cup \bm{s}$, $ \bm{t} \cup \bm{j}$,  and  $\bm{t} \cup \bm{l}$ analogously.  We also define sequences of {\it double} indices,
\begin{align*}
\bm{i} \cup \bm{s} | \bm{t} \cup \bm{j} &= (i_1 t_1, \hdots, i_p t_p, s_1 j_1, \hdots, s_q j_q), \\
\bm{k} \cup \bm{s} | \bm{t} \cup \bm{l} &= (k_1 t_1, \hdots, k_p t_p, s_1 l_1, \hdots, s_q l_q).
\end{align*}
From \eqref{eq:proj_pi} we have
\begin{multline} \label{eqn:exp-linearity}
\mathbb{E} \left[ 
\prod_{\alpha=1}^p \pi_1 (H)_{i_\alpha, k_\alpha} \cdot \prod_{\beta=1}^{q} \pi_2 (H)_{j_\beta, l_\beta} 
\right] \\
= \sum_{t_1, \dots, t_p=1}^n  \sum_{s_{1}, \dots, s_{q}=1}^m
\mathbb{E}\left[ 
\prod_{\alpha=1}^p H_{i_\alpha t_\alpha, k_\alpha t_\alpha} \cdot \prod_{\beta=1}^{q} H_{s_\beta j_\beta, s_\beta l_\beta} 
\right].
\end{multline}
Note that in each term in the sum above, the row index sequence in the quantity in brackets
$$\prod_{\alpha=1}^p H_{i_\alpha t_\alpha, k_\alpha t_\alpha} \cdot \prod_{\beta=1}^{q} H_{s_\beta j_\beta, s_\beta l_\beta}$$
is  $\bm{i} \cup \bm{s} | \bm{t} \cup \bm{j}$, while the column index sequence is  $\bm{k} \cup \bm{s} | \bm{t} \cup \bm{l}$.  Applying Lemma \ref{lem:CMS} to each term, the right-hand side of (\ref{eqn:exp-linearity}) then becomes
\[
\sum_{t_1, \dots, t_p=1}^n  \sum_{s_{1}, \dots, s_{q}=1}^m
\sum_{\sigma, \tau \in \mathfrak{S}_{p+q}} 
\delta_{\sigma} (\bm{i}\cup \bm{s}, \bm{k}\cup \bm{s}) \, 
\delta_{\sigma} ( \bm{t} \cup \bm{j}, \bm{t} \cup \bm{l}) 
\Wg_{mn} (\sigma^{-1}\tau) \Tr_\tau(\lambda).
\]
Changing the order of sums, the above expression is equal to
\begin{multline*}
\sum_{\sigma, \tau \in \mathfrak{S}_{p+q}}
\left[ \sum_{s_{1}, \dots, s_{q}=1}^m
\delta_{\sigma} (\bm{i}\cup \bm{s}, \bm{k}\cup \bm{s}) \right] 
\left[  \sum_{t_1, \dots, t_p=1}^n  
\delta_{\sigma} (\bm{t} \cup \bm{j} , \bm{t} \cup \bm{l})  \right]  \\
 \qquad \qquad \times \Wg_{mn} (\sigma^{-1}\tau) \Tr_\tau(\lambda).
\end{multline*}
Therefore Theorem \ref{thm:projection_distribution} follows from the lemma below.

\begin{lem}
Let $\bm{i}, \bm{j},\bm{k},\bm{l}$ be as in Theorem \ref{thm:projection_distribution}.
For each $\sigma \in \mathfrak{S}_{p+q}$, we have
\begin{align}
\sum_{s_{1}, \dots, s_{q}=1}^m
\delta_{\sigma} (\bm{i}\cup \bm{s}, \bm{k}\cup \bm{s}) 
&= \delta_{\pr_1 (\sigma)}(\bm{i},\bm{k}) m^{\kappa_1(\sigma)}, 
\label{eq:key-lemma-1}\\
\sum_{t_1, \dots, t_p=1}^n  
\delta_{\sigma} (\bm{t} \cup \bm{j} , \bm{t} \cup \bm{l}) 
&= \delta_{\pr_2 (\sigma)} (\bm{j},\bm{l}) n^{\kappa_2(\sigma)}.
\end{align}
\end{lem}

\begin{proof}
By symmetry, it is enough to show the first equation
\eqref{eq:key-lemma-1}.
From \eqref{eq:delta-cycle-decomposition}, we obtain the decomposition
\begin{equation} \label{eq:sum-s-decomposition}
\sum_{s_{1}, \dots, s_{q}=1}^m
\delta_{\sigma} (\bm{i}\cup \bm{s}, \bm{k}\cup \bm{s}) 
= \prod_{c \in C(\sigma)} \left(
\sum_{(s_\beta): p+\beta \in c} 
\delta_{c} (\bm{i}\cup \bm{s}, \bm{k}\cup \bm{s}) 
 \right),
\end{equation}
where $C(\sigma)$ is the set of all cycles in $\sigma$,
and the sum on the right-hand side is over all $s_\beta$ such that 
the letter $p+\beta$ appears in $c$ and $\beta \in \{1,2,\dots,q\}$.
For example, for $p=q=5$ and $\sigma$ given in \eqref{eq:example_permutation}, we have
\begin{align*}
&\sum_{s_{1}, \dots, s_{5}=1}^m
\delta_{\sigma} (\bm{i}\cup \bm{s}, \bm{k}\cup \bm{s}) \\
&= \delta_{(1 \ 3)} (\bm{a},\bm{b})\cdot 
\delta_{(5)} (\bm{a},\bm{b})  \cdot
\sum_{s_2, s_1, s_4=1}^m \delta_{(2 \ 7 \ 6 \ 4 \ 9)} (\bm{a},\bm{b}) \cdot \sum_{s_3,s_5=1}^m \delta_{(8 \ 10)} (\bm{a},\bm{b}).
\end{align*}
(Here $s_2, s_1,s_4$ appear for the cycle $c=(2 \ 7 \ 6 \ 4 \ 9)$
because $p+2, p+1, p+4$ with $p=5$ appear in $c$.)

Now we compute $$\sum\limits_{(s_\beta): p+\beta \in c} 
\delta_{c} (\bm{i}\cup \bm{s}, \bm{k}\cup \bm{s})$$
for each cycle $c$ in $\sigma$.
There are three possible cases.
\begin{enumerate}
\item[(i)] Suppose that $c \subset \{1,2,\dots, p\}$.
Then there is no $\beta$ such that $p+\beta \in c$. 
We therefore have 
\[
\sum\limits_{(s_\beta): p+\beta \in c} 
\delta_{c} (\bm{i}\cup \bm{s}, \bm{k}\cup \bm{s}) 
= \delta_{c} (\bm{i}\cup \bm{s}, \bm{k}\cup \bm{s}) 
= \delta_c (\bm{i}, \bm{k}).
\]
\item[(ii)] Suppose that $c \subset \{p+1,p+2,\dots, p+q\}$
and write $$c=(p+\beta_1 \ \ p+\beta_2 \ \ \dots \ \  p+\beta_h).$$
If $h=1$ (i.e. $p+\beta_1$ is a fixed point of $\sigma)$, then 
\[
\sum\limits_{(s_\beta): p+\beta \in c} 
\delta_{c} (\bm{i}\cup \bm{s}, \bm{k}\cup \bm{s}) 
=\sum_{s_{\beta_1}=1}^m \delta_{(p+\beta_1)}  (\bm{i}\cup \bm{s}, \bm{k}\cup \bm{s}) 
= \sum_{s_{\beta_1}=1}^m  \delta(s_{\beta_1},s_{\beta_1})= m.
\]
If $h >1$, then
\begin{align*}
&\sum\limits_{(s_\beta): p+\beta \in c} 
\delta_{c} (\bm{i}\cup \bm{s}, \bm{k}\cup \bm{s}) 
=\sum_{s_{\beta_1},s_{\beta_2}  \dots, s_{\beta_h}=1}^m
\delta_{c}  (\bm{i}\cup \bm{s}, \bm{k}\cup \bm{s}) \\
&= \sum_{s_{\beta_1},s_{\beta_2}  \dots, s_{\beta_h}=1}^m
\delta(s_{\beta_2}, s_{\beta_1}) 
\cdots \delta(s_{\beta_{h}}, s_{\beta_{h-1}}) \, 
\delta(s_{\beta_1}, s_{\beta_h})
=m.
\end{align*}
\item [(iii)] Suppose that $c \cap \{1,2,\dots,p\} \not=\emptyset$ and also $c \cap \{p+1,p+2,\dots, p+q\} \not= \emptyset$.
Then
\begin{equation} \label{eq:case-iii}
\sum\limits_{(s_\beta): p+\beta \in c} 
\delta_{c} (\bm{i}\cup \bm{s}, \bm{k}\cup \bm{s}) 
= \delta_{\pr_1(c)} (\bm{i}, \bm{k}),
\end{equation}
where we extend the definition of $\pr_1:\mathfrak{S}_{p+q} \to \mathfrak{S}_p$ to cycles on $\{1,2,\dots,p+q\}$ in the natural way.
For any cycle $c$ in $\sigma$, 
the image $\pr_1(c)$ is a cycle of $\pr_1(\sigma)$.

An illustrative example serves to show why (\ref{eq:case-iii}) must hold in this case.
Let $p=q=5$ and $c=(2 \ 7 \ 6 \ 4 \ 9)=(2 \ \  p+2 \ \ p+1  \ \ 4 \ \ p+4)$. Then,
by the definition of $\delta_{c} (\bm{i}\cup \bm{s}, \bm{k}\cup \bm{s})$, we can see that
\begin{align*}
&\sum\limits_{(s_\beta): p+\beta \in c} 
\delta_{c} (\bm{i}\cup \bm{s}, \bm{k}\cup \bm{s}) \\
&= \sum_{s_2, s_1,s_4=1}^m 
\delta(s_2, k_2) \, \delta(s_1,s_2) \, \delta(i_4, s_1) \,
\delta(s_4,k_4) \, \delta(i_2, s_4).
\end{align*}
Only one term with $k_2=s_2=s_1=i_4$ and $k_4=s_4=i_2$ survives.
This survivor corresponds to the elimination 
\[
\Big( \ \underbrace{2}_{i_2/ k_2} \ \  \underbrace{\cancel{p+2}}_{s_2} \ \ \underbrace{\cancel{p+1}}_{s_1}  \ \ \underbrace{4}_{i_4/k_4} \ \ \underbrace{\cancel{p+4}}_{s_4} \ \Big)
\] 
for $c$.
Hence
the above summation equals $\delta(i_4,k_2) \delta(i_2, k_4) =\delta_{(2 \ 4)}(\bm{i},\bm{k})$.
\end{enumerate}

Returning to \eqref{eq:sum-s-decomposition} and using the formulae for the three cases above, we have
\begin{align*}
& \prod_{c \in C(\sigma)} \left(
\sum_{(s_\beta): p+\beta \in c} 
\delta_{c} (\bm{i}\cup \bm{s}, \bm{k}\cup \bm{s}) 
 \right) \\
&= \prod_{c: \, \text{case (i)}}  \delta_c (\bm{i}, \bm{k})
\times \prod_{c: \, \text{case (iii)}}  \delta_{\pr_1(c)} (\bm{i}, \bm{k}) \times  \prod_{c: \, \text{case (ii)}} m.
\end{align*}
The number of cycles in case (ii) is $\kappa_1(\sigma)$ by definition.
It is easy to check that 
\[
\pr_1(\sigma)= 
\prod_{c: \, \text{case (i)}} c \cdot \prod_{c: \, \text{case (iii)}}
\pr_1(c).
\]
Hence we have completed the proof of \eqref{eq:key-lemma-1}.
\end{proof}

\subsection{More than two distinguishable particles}

An immediate generalization of Theorem \ref{thm:projection_distribution} gives a formula for the mixed moments of entries of single-particle marginals for a system of $k$ distinguishable particles.  Here we choose integers $n_1, \hdots, n_k \ge 2$ representing the dimensions of the Hilbert spaces for each of the individual particles, and we set $N = \prod_{j=1}^k n_j$.
Identifying
\begin{equation} \label{eq:CN-tensor}
\bC^N \cong \bigotimes_{j = 1}^k \bC^{n_j},
\end{equation}
 we label the coordinates of $N$-dimensional vectors by $k$-component multi-indices $({i_1, \hdots, i_k})$, where each $i_j$ runs from 1 to $n_j$ and the multi-indices are ordered lexicographically.
Fix $\lambda =(\lambda_{i_1 \hdots i_k})_{i_1 \hdots i_k=1 \hdots 1}^{n_1 \hdots n_k} \in \mathbb{R}^{N}$,
and let $U$ be a Haar-distributed $N$-by-$N$ unitary matrix.
We consider a random Hermitian matrix given by 
\begin{equation}
H=(H_{i_1 \hdots i_k, j_1 \hdots j_k}) = U \diag (\lambda) U^\dagger.
\end{equation}

For $l = 1, \hdots, k$, define the $n_l \times n_l$ Hermitian matrix $\pi_l(H) =(\pi_l(H)_{i_l,j_l})$ by
\begin{equation} \label{eq:proj_multi_1}
\pi_l(H)_{i_l,j_l}= \sum_{\substack{1 \le s \le k \\ s \ne l}} \sum_{i_s=1}^{n_s} H_{i_1 \hdots i_l \hdots i_k, i_1 \hdots j_l \hdots i_k},
\end{equation}
that is, $\pi_l(H)$ is obtained by taking the partial trace of $H$ over all legs of the tensor product (\ref{eq:CN-tensor}) {\it except} $\bC^{n_l}$.  Fix nonnegative integers $p_1, \hdots, p_k$.  For $\sigma \in \mathfrak{S}_{p_1+ \cdots + p_k}$, we make the following definitions, which are analogous to those in \S \ref{subsec:sym_gp}:
\begin{itemize}
\item $\kappa_l(\sigma)$ is the number of cycles $c$ in $\sigma$ satisfying $$c \cap \bigg \{1 + \sum_{j =1}^{l-1} p_{j}, \, \hdots, \, \sum_{j =1}^{l} p_{j} \bigg \} = \emptyset.$$
\item $\pr_l(\sigma) \in \mathfrak{S}_{p_l}$ is the permutation on $p_l$ elements obtained by decomposing $\sigma$ as a product of cycles, erasing all elements except for $i + \sum_{j =1}^{l-1} p_{j}$ for $i = 1, \hdots, p_l$, and then replacing each of the remaining elements $i + \sum_{j =1}^{l-1} p_{j}$ with $i$.
\end{itemize}
An identical argument to the proof of Theorem  \ref{thm:projection_distribution} yields the following formula for the mixed moment 
\begin{equation}
\mathbb{E} \left[ 
\prod_{l = 1}^k \prod_{\alpha=1}^{p_l} \pi_l(H)_{i^{(l)}_\alpha, j^{(l)}_\alpha}
\right]
\end{equation}
for $i^{(l)}_1,\dots,i^{(l)}_{p_l} , j^{(l)}_1,\dots,j^{(l)}_{p_l} \in \{1,2,\dots,n_l\}$:

\begin{thm} \label{thm:projection_multi_1}
Fix $k \ge 2$ and let $p_1, \hdots, p_k$ be nonnegative integers. For $l = 1, \hdots, k$, fix sequences of indices
\[
    \bm{i}^{(l)}=(i^{(l)}_1,\dots,i^{(l)}_{p_l}), \ \bm{j}^{(l)}=(j^{(l)}_1,\dots,j^{(l)}_{p_l})  \quad \in \{1,2,\dots,n_l\}^{\times p_l}.
\]
We have the formula
\begin{multline}
\mathbb{E} \left[ 
\prod_{l = 1}^k \prod_{\alpha=1}^{p_l} \pi_l(H)_{i^{(l)}_\alpha, j^{(l)}_\alpha}
\right] \\
= \sum_{\sigma ,\tau \in \mathfrak{S}_{p_1 + \cdots + p_k}} 
\Wg_{N}(\sigma^{-1}\tau) \Tr_\tau(\lambda)
 \prod_{l = 1}^k \delta_{\pr_l (\sigma)} (\bm{i}^{(l)},\bm{j}^{(l)}) \, n_l^{\kappa_l(\sigma)}.
\end{multline}
\end{thm}

\section{Bosons and fermions} \label{sec:bos-fer}

\subsection{Bosons}

\textit{Indistinguishable bosons} are particles whose joint state is left invariant when any two of the particles are exchanged.  Accordingly, we consider a ``state'' of a system of $k$ indistinguishable bosons to be a Hermitian operator on $\mathrm{Sym}^k \bC^n$, where $\bC^n$ represents the Hilbert space associated with an individual boson.  By a standard construction, described in more detail in \cite[\textsection2.3]{CM}, we model $\mathrm{Sym}^k \bC^n$ as the space of homogenerous polynomials of degree $k$ in $n$ variables $(x_1, \hdots, x_n)$.  We make this space into a Hilbert space by choosing as an orthonormal basis the normalized monomials
\[
    v_\alpha = \frac{1}{\sqrt{\alpha!}} \prod_{i=1}^n x_i^{\alpha_i},
\]
as $\alpha$ runs over $n$-component multi-indices with $|\alpha| = \sum_i \alpha_i = k$, and the multi-index factorial is defined by $\alpha! = \prod_i \alpha_i !$.  As there are $N = {{n + k -1} \choose {k}}$ such multi-indices, this allows us to identify $\mathrm{Sym}^k \bC^n \cong \bC^N$.  The symmetric group $\mathfrak{S}_k$ acts on the tensor product space $(\bC^n)^{\otimes k} \cong \bC^{n^k}$ by permuting the legs of the tensor product, and there is an isometric embedding $S : \mathrm{Sym}^k \bC^n \hookrightarrow (\bC^n)^{\otimes k}$ given by
\begin{equation} \label{eq:sym_embed}
 S : v_\alpha \mapsto \frac{1}{\sqrt{k! \, \alpha!}} \sum_{\sigma \in \mathfrak{S}_k} \sigma( e_1^{\otimes \alpha_{1}} \otimes \cdots \otimes e_n^{\otimes \alpha_n} ),
 \end{equation}
where $e_1, \hdots, e_n$ are the standard basis vectors of $\bC^n$, and we use the convention $e_i \otimes e_j^{\otimes 0} = e_i$.  To verfify that (\ref{eq:sym_embed}) is an isometry, note that the stabilizer of $e_1^{\otimes \alpha_{1}} \otimes \cdots \otimes e_n^{\otimes \alpha_n}$ in $\mathfrak{S}_k$ has order $\alpha!$.  Therefore the sum in (\ref{eq:sym_embed}) includes $k!/\alpha!$ distinct pairwise orthogonal vectors of length $\alpha!$, so that its magnitude is equal to $\sqrt{(k! / \alpha!) \alpha!^2 } = \sqrt{k! \, \alpha!}.$ Thus $S v_\alpha$ is indeed a unit vector in $(\bC^n)^{\otimes k}$.  It is obvious that if $\alpha \ne \beta$ then $S v_\alpha$ and $S v_\beta$ are orthogonal.

Using (\ref{eq:sym_embed}), we can write down the matrix elements of $S$ with respect to the basis $\{ v_\alpha \}$ of $\mathrm{Sym}^k \bC^n$ and the standard basis
$$ e_{\otimes \bm{i}} = e_{i_1} \otimes \cdots \otimes e_{i_k}, \qquad \bm{i} = (i_1, \hdots, i_k) \in \{ 1, \hdots, n \}^{\times k} $$
of $(\bC^n)^{\otimes k}$.  For $\bm{i} \in \{ 1, \hdots, n \}^{\times k}$, write $\mathrm{tab}(\bm{i})$ for the $n$-component multi-index such that, for each $j \in \{1, \hdots, n\}$, the $j$th component $\mathrm{tab}(\bm{i})_j$ is equal to the multiplicity with which $j$ occurs in $\bm{i}$.  Then the entries of $S$, regarded as an $n^k$-by-$N$ matrix, are
\begin{equation} \label{eq:S-entries}
S_{\bm{i},\alpha} =  \sqrt{\frac{\alpha!}{k!}} \, \delta(\mathrm{tab}(\bm{i}), \alpha),
\end{equation}
where $\delta(\mathrm{tab}(\bm{i}), \alpha) = 1$ if $\mathrm{tab}(\bm{i}) = \alpha$ and $0$ otherwise.  This matrix has an $N$-by-$n^k$ pseudoinverse $(\bC^n)^{\otimes k} \to \mathrm{Sym}^k \bC^n$, which is obtained by first projecting onto the span of the vectors $S v_\alpha$, and then mapping $S v_\alpha \mapsto v_\alpha$.  In fact, since $S$ is an isometry with real entries, this pseudoinverse is just the transpose $S^T$; it is easily verified from (\ref{eq:sym_embed}) and (\ref{eq:S-entries}) that $S^T S v_\alpha = v_\alpha$.

The map $H \mapsto S H S^T$ embeds the space of $N$-by-$N$ Hermitian matrices into the space of $n^k$-by-$n^k$ Hermitian matrices.  Our main objects of study in this section are the single-particle marginals of the state $H$, which are the $n$-by-$n$ matrices $\pi_i(SHS^T)$, $1 \le i \le n$, where $\pi_i$ is the $i$th marginal of a Hermitian operator on $(\bC^n)^{\otimes k}$ as previously defined in (\ref{eq:proj_multi_1}).  However, since $\mathrm{span}\{S v_\alpha \} \cong \mathrm{Sym}^k \bC^n$ is invariant under the action of $\mathfrak{S}_k$, we have $\pi_i(SHS^T) = \pi_j(SHS^T)$ for all $1 \le i, j \le k$.  Without loss of generality, we therefore restrict our attention to $\pi(H) = \pi_1(SHS^T)$.

We pose the following problem.  Let $\lambda = (\lambda_\alpha)_{|\alpha| = k} \in \bR^N$, and let $U$ be a Haar-distributed random $N$-by-$N$ unitary matrix.  Set $H = U \mathrm{diag}(\lambda) U^\dagger$.  Fix a nonnegative integer $p$ and two sequences of indices $(i_1, \hdots, i_p),$ $(j_1, \hdots, j_p) \in \{1, \hdots, n\}^{\times p}$.  We want to derive a formula for the mixed moment
$$ \mathbb{E} \left[ \prod_{k = 1}^p \pi(H)_{i_k, j_k} \right]. $$

In the next lemma we write the entries of $\pi(H)$ in terms of the entries of $H$, which will enable us to solve this problem using Weingarten calculus.  First we introduce some further notation.  Given two $n$-component multi-indices $\alpha, \beta$, we define their sum $\alpha + \beta$ componentwise, as though $\alpha$ and $\beta$ were vectors in $\bR^n$: $(\alpha + \beta)_j = \alpha_j + \beta_j$.  We write $e_j$ for the multi-index with a $1$ in the $j$th position and all other entries equal to $0$.  Thus $(\alpha + e_j)_j = \alpha_j + 1$, and $(\alpha + e_j)_i = \alpha_i$ for $i \ne j$.

\begin{lem} \label{lem:proj_entries_bos}
\begin{equation} \label{eq:proj_entries_bos}
    \pi(H)_{ij} = \frac{1}{k} \sum_{|\gamma| = k-1} \sqrt{(\gamma_i + 1)(\gamma_j + 1)} \, H_{e_i + \gamma, e_j + \gamma},
\end{equation}
where the sum runs over $n$-component multi-indices $\gamma$ satisfying $|\gamma| = k-1$.
\end{lem}

\begin{proof}
As a first step, we compute the entries of $S H S^T$.  For $\bm{i}, \bm{j} \in \{1, \hdots, n\}^{\times k}$, (\ref{eq:S-entries}) gives
\[
    (S H S^T)_{\bm{i}, \bm{j}}
        = \sum_{|\alpha| = k} S_{\bm{i}, \alpha} \sum_{|\beta| = k} H_{\alpha, \beta} S^T_{\beta, \bm{j}} 
        = \frac{\sqrt{ \mathrm{tab}(\bm{i})! \, \mathrm{tab}(\bm{j})!} }{k!} H_{\mathrm{tab}(\bm{i}), \mathrm{tab}(\bm{j})}.
\]
Then from the definition (\ref{eq:proj_multi_1}), we have
$$ \pi(H)_{ij} = \pi_1(S H S^T)_{ij} = \sum_{l_1, \hdots, l_{k-1} = 1}^n (SHS^T)_{i l_1 \hdots l_{k-1}, \, j l_1 \hdots l_{k-1}} $$
$$ = \sum_{l_1, \hdots, l_{k-1} = 1}^n \frac{\sqrt{ \mathrm{tab}(i, l_1, \hdots, l_{k-1})! \, \mathrm{tab}(j, l_1, \hdots, l_{k-1})!} }{k!} H_{\mathrm{tab}(i, l_1, \hdots, l_{k-1}), \mathrm{tab}(j, l_1, \hdots, l_{k-1})}.$$
Clearly we can write $\mathrm{tab}(i, l_1, \hdots, l_{k-1}) = e_i + \gamma$, $\mathrm{tab}(j, l_1, \hdots, l_{k-1}) = e_j + \gamma$ for a unique multi-index $\gamma$ with $|\gamma| = k-1$.  Namely, $\gamma = \mathrm{tab}(l_1, \hdots, l_{k-1}).$  Each such $\gamma$ appears in the sum above with multiplicity equal to the number of distinct permutations of any given $(k-1)$-tuple $(l_1, \hdots, l_{k-1})$ with $\mathrm{tab}(l_1, \hdots, l_{k-1}) = \gamma$.  The stabilizer of such a $(k-1)$-tuple in $\mathfrak{S}_{k-1}$ has order $\gamma!$, so there are $(k-1)! / \gamma!$ distinct permutations. Thus we find
$$ \pi(H)_{ij} = \sum_{|\gamma| = k-1} \frac{(k-1)!}{\gamma!} \frac{\sqrt{(e_i + \gamma)! \, (e_j + \gamma)!}}{k!}  \, H_{e_i + \gamma, e_j + \gamma}, $$
from which (\ref{eq:proj_entries_bos}) follows by the observation that $(e_i + \gamma)! = \gamma! (\gamma_i + 1)$.
\end{proof}

Combining Lemmas \ref{lem:CMS} and \ref{lem:proj_entries_bos}, we have shown the desired formula:
\begin{thm} \label{thm:bos-moments}
The mixed moments of the entries of a single-particle marginal for a system of $k$ indistinguishable bosons are given by
\begin{equation} \label{eq:boson_proto}
\mathbb{E} \left[ 
\prod_{s=1}^p \pi (H)_{i_s, j_s} 
\right]
= \sum_{\sigma, \tau \in \mathfrak{S}_p} \Wg_N( \sigma^{-1}\tau) \Tr_\tau(\lambda)
\Delta_{\sigma}^{n,k}(\bm{i}, \bm{j})
\end{equation} 
with
\begin{multline} \label{eqn:Delta-bos}
\Delta_{\sigma}^{n,k}(\bm{i}, \bm{j})
=  \frac{1}{k^p} \sum_{\gamma^{(1)}, \dots, \gamma^{(p)}}\prod_{s=1}^p \sqrt{(\gamma_{i_s}^{(s)}+1)(\gamma_{j_s}^{(s)}+1)} \\
\times \delta_{\sigma} \big( (e_{i_1} +\gamma^{(1)}, \dots, e_{i_p}+\gamma^{(p)}),
(e_{j_1} +\gamma^{(1)}, \dots, e_{j_p}+\gamma^{(p)})
\big),
\end{multline}
where each $\gamma^{(s)}$ runs over $n$-component multi-indices such that 
the sum of their components is $k-1$.
\end{thm}

\subsubsection{Bosons: A simple example}

Consider the case $p=1$. 
Then, for the unique permutation $\mathrm{id}_1 \in \mathfrak{S}_1$,
\[
\Delta_{\mathrm{id}_1}^{n,k} (i,j)= \frac{1}{k} \sum_{|\gamma|=k-1} \sqrt{(\gamma_{i}+1)(\gamma_j+1)} 
\, \delta_{\mathrm{id}_1} (e_i+\gamma, e_j+\gamma).
\]
It is easy to see that this equals
\[
\Delta_{\mathrm{id}_1}^{n,k} (i,j)= \frac{\delta(i, j)}{k}\sum_{\gamma_1+ \cdots+\gamma_n=k-1} (\gamma_{i}+1).
\]
By symmetry, we have
\begin{align*}
\sum_{\gamma_1+ \cdots+\gamma_n=k-1} (\gamma_{i}+1)
&=\sum_{\gamma_1+ \cdots+\gamma_n=k-1} (\gamma_{1}+1) \\ 
&= \sum_{a=0}^{k-1} (a+1) \sum_{\gamma_2+ \cdots+ \gamma_{n}= k-1-a} 1 \\
&= \sum_{a=0}^{k-1} (a+1)  \binom{n-1+(k-1-a)-1}{k-1-a}.
\end{align*}
One can check (e.g.~using Mathematica) that this equals $\binom{n+k-1}{k-1}$.
Therefore we have obtained 
\begin{equation}
\Delta_{\mathrm{id}_1}^{n,k} (i,j)= \delta(i,j) \frac{1}{k} \binom{n+k-1}{k-1}
\end{equation}
for $p=1$.

From \eqref{eq:boson_proto}, we now have
\[
\mathbb{E} \left[\pi(H)_{i,j} \right]= \Wg_N(\mathrm{id}_1) \Tr_{\mathrm{id}_1}(\lambda) 
\, \Delta_{\mathrm{id}_1}^{n,k} (i,j) 
= \frac{1}{N} \Tr(\lambda)  \, \delta(i,j) \frac{1}{k} \binom{n+k-1}{k-1}.
\]
With $N=\binom{n+k-1}{k}$, it is immediate to see that
\begin{equation}
\mathbb{E} \left[\pi(H)_{i,j} \right] = \Tr(\lambda) \frac{\delta(i,j)}{n},
\end{equation}
exactly as one would expect.

\subsection{Fermions}
\textit{Indistinguishable fermions} are particles whose joint state is \textit{anti}-invariant (i.e., changes sign) when any two of the particles are exchanged.  Accordingly, we consider a ``state'' of a system of $k$ indistinguishable fermions to be a Hermitian operator on $\wedge^k \bC^n$, where $\bC^n$ represents the Hilbert space associated with an individual fermion. As explained in \cite[\textsection2.4]{CM}, we assume that $1 < k < n-1$, as otherwise the quantum marginal problem for fermions is trivial.  Define
$$ \mathcal{A}_k = \big \{ (a_1, \hdots, a_k) \in \{1, \hdots, n\}^{\times k} \ \big | \ a_1 < \hdots < a_k \big \}.$$
We make $\wedge^k \bC^n$ into a Hilbert space by choosing as an orthonormal basis the ${{n} \choose {k}}$ $k$-vectors
$$ e_{\wedge \bm{a}} = e_{a_1} \wedge \cdots \wedge e_{a_k}, \qquad \bm{a} = (a_1, \hdots, a_k) \in \mathcal{A}_k. $$
Then we have an isometric embedding $A : \wedge^k \bC^n \hookrightarrow ( \bC^n )^{\otimes k}$ defined by
$$ A e_{\wedge \bm{a}} = \frac{1}{\sqrt{k!}} \sum_{\sigma \in \mathfrak{S}_k} \sgn(\sigma) \, \sigma( e_{a_1} \otimes \cdots \otimes e_{a_k} ), $$
where $\mathfrak{S}_k$ acts on $( \bC^n )^{\otimes k}$ by permuting the legs of the tensor product.  For $\bm{i} \in \{1, \hdots, n\}^{\times k}$, set $\sgn(\bm{i}) = 0$ if not all entries of $\bm{i}$ are distinct; otherwise, set $\sgn(\bm{i}) = \sgn(\sigma)$, where $\sigma \in \mathfrak{S}_k$ is the unique permutation such that $\sigma(\bm{i})_1 < \hdots < \sigma(\bm{i})_k$.  Then the matrix entries of $A$ with respect to the basis $\{ e_{\wedge \bm{a}} \}_{\bm{a} \in \mathcal{A}_k}$ of $\wedge^k \bC^n$ and the basis $\{ e_{\otimes \bm{i}} \}_{\bm{i} \in \{1, \hdots, n\}^{\times k}}$ of $( \bC^n )^{\otimes k}$ are
$$ A_{\bm{i}, \bm{a}} = \frac{\sgn(\bm{i})}{\sqrt{k!}} \delta(\mathrm{sr}(\bm{i}), \bm{a}), $$
where the map $\mathrm{sr}$ sorts the entries of a vector in nondecreasing order, that is,
$$ \{ \mathrm{sr}(\bm{i})_1, \hdots, \mathrm{sr}(\bm{i})_k \} = \{ i_1, \hdots, i_k \}, \qquad \mathrm{sr}(\bm{i})_1 \le \hdots \le \mathrm{sr}(\bm{i})_k.$$
The transpose $A^T$ is a pseudoinverse such that $A^T A e_{\wedge \bm{a}} = e_{\wedge \bm{a}}$, and the map $H \mapsto AHA^T$ embeds the space of ${{n} \choose {k}}$-by-${{n} \choose {k}}$ Hermitian matrices into the space of $n^k$-by-$n^k$ Hermitian matrices.  For an  ${{n} \choose {k}}$-by-${{n} \choose {k}}$ Hermitian matrix $H$, define $\pi(H) = \pi_1(AHA^T).$

We can now state the quantum marginal problem for $k$ indistinguishable fermions.  Let $\lambda = (\lambda_{\bm{a}})_{\bm{a} \in \mathcal{A}_k} \in \bR^{{n} \choose {k}}$, and let $U$ be a Haar-distributed random ${{n} \choose {k}}$-by-${{n} \choose {k}}$ unitary matrix.  Set $H = U \mathrm{diag}(\lambda) U^\dagger$.  Fix a nonnegative integer $p$ and two sequences of indices $(i_1, \hdots, i_p),$ $(j_1, \hdots, j_p) \in \{1, \hdots, n\}^{\times p}$.  We want to derive a formula for the mixed moment
$$ \mathbb{E} \left[ \prod_{k = 1}^p \pi(H)_{i_k, j_k} \right]. $$

To solve this problem using Weingarten calculus, we will need the following expression for the entries of $\pi(H)$ in terms of the entries of $H$.

\begin{lem} \label{lem:proj_entries_fer}
\begin{equation} \label{eq:proj_entries_fer}
\pi(H)_{ij} = \frac{1}{k} \sum_{\substack{ \bm{l} \in \mathcal{A}_{k-1} \\ \{l_1, \hdots, l_{k-1}\} \not \ni i, j }} \sgn(i, \bm{l} ) \, \sgn(j, \bm{l} ) \, H_{\mathrm{sr}(i, \bm{l} ), \mathrm{sr}(j,\bm{l} )},
\end{equation}
where $\sgn(i, \bm{l} )$ means $\sgn(i, l_1, \hdots, l_{k-1} )$ and $\mathrm{sr}(i, \bm{l} )$ means $\mathrm{sr}(i, l_1, \hdots, l_{k-1} )$.
\end{lem}

\begin{proof}
We first compute
$$ (AHA^T)_{\bm{i}, \bm{j}} = \sum_{\bm{a} \in \mathcal{A}_k} A_{\bm{i}, \bm{a}} \sum_{\bm{b} \in \mathcal{A}_k} H_{\bm{a}, \bm{b}} A^T_{\bm{b}, \bm{j}} =  \frac{\sgn(\bm{i}) \, \sgn(\bm{j})}{k!} H_{\mathrm{sr}(\bm{i}), \mathrm{sr}(\bm{j})},$$
where we set $H_{\mathrm{sr}(\bm{i}), \mathrm{sr}(\bm{j})} = 0$ if $\bm{i}$ or $\bm{j}$ has repeated entries.  Then we have
\begin{align*}
\pi(H)_{ij} &= \sum_{l_1, \hdots, l_{k-1} = 1}^n (AHA^T)_{(i, \bm{l}), (j, \bm{l})} = \sum_{l_1, \hdots, l_{k-1} = 1}^n \frac{\sgn(i, \bm{l}) \, \sgn(j, \bm{l})}{k!} H_{\mathrm{sr}(i, \bm{l}), \mathrm{sr}(j, \bm{l})} \\
&= \frac{1}{k!} \sum_{\substack{ \bm{l} \in \mathcal{A}_{k-1} \\ \{l_1, \hdots, l_{k-1}\} \not \ni i, j }} H_{\mathrm{sr}(i, \bm{l}), \mathrm{sr}(j, \bm{l})} \sum_{\sigma \in \mathfrak{S}_{k-1}} \sgn(i, \sigma(\bm{l}) ) \, \sgn(j, \sigma(\bm{l})).
\end{align*}
Observing that
$$\sgn(i, \sigma(\bm{l}) ) \, \sgn(j, \sigma(\bm{l})) =  \sgn(i, \bm{l} ) \, \sgn(j, \bm{l}) \, \sgn(\sigma)^2 =  \sgn(i, \bm{l} ) \, \sgn(j, \bm{l})$$
for all $\sigma \in \mathfrak{S}_{k-1}$, we obtain (\ref{eq:proj_entries_fer}).
\end{proof}

Similarly to the bosonic case, from Lemmas \ref{lem:CMS} and \ref{lem:proj_entries_fer}, we obtain the desired formula:

\begin{thm} \label{thm:fer-moments}
The mixed moments of the entries of a single-particle marginal for a system of $k$ indistinguishable fermions are given by
\begin{equation*} 
\mathbb{E} \left[ 
\prod_{s=1}^p \pi (H)_{i_s, j_s} 
\right]
= \sum_{\sigma, \tau \in \mathfrak{S}_p} \Wg_N( \sigma^{-1}\tau) \Tr_\tau(\lambda)
\Delta_{\sigma}^{n,k}(\bm{i}, \bm{j})
\end{equation*} 
with
\begin{multline} \label{eqn:Delta-fer}
\Delta_{\sigma}^{n,k}(\bm{i}, \bm{j})
=  \frac{1}{k^p}
\sum_{\substack{ \bm{l}^{(1)} \in \mathcal{A}_{k-1} \\ \bm{l}^{(1)} \not \ni i_1, j_1 }} \cdots
\sum_{\substack{ \bm{l}^{(p)} \in \mathcal{A}_{k-1} \\ \bm{l}^{(p)} \not \ni i_p, j_p }}
\prod_{s=1}^p \{\sgn(i_s, \bm{l}^{(s)}) \sgn(j_s, \bm{l}^{(s)})\}
 \\
\times \delta_{\sigma} \Big( \big( \mathrm{sr}(i_1,\bm{l}^{(1)}), \dots ,
 \mathrm{sr}(i_p,\bm{l}^{(p)}) \big), 
 \big( \mathrm{sr}(j_1,\bm{l}^{(1)}), \dots ,
 \mathrm{sr}(j_p,\bm{l}^{(p)}) \big) \Big).
\end{multline}
\end{thm}

This $\Delta$-function may be decomposed into the cycles of $\sigma$ as in \eqref{eq:delta-cycle-decomposition}.

\subsubsection{Fermions: A simple example}

Again consider the case $p=1$. 
Then, for the unique permutation $\mathrm{id}_1 \in \mathfrak{S}_1$,
\begin{align*}
\Delta_{\mathrm{id}_1}^{n,k} (i,j)
&=
\frac{1}{k}
\sum_{\substack{ \bm{l} \in \mathcal{A}_{k-1} \\ \bm{l} \not \ni i, j }} 
 \sgn(i, \bm{l}) \sgn(j, \bm{l})
 \delta \big( \mathrm{sr}(i,\bm{l}), \mathrm{sr}(j,\bm{l}) \big) \\
&= \delta(i,j) \frac{1}{k} \sum_{\substack{ \bm{l} \in \mathcal{A}_{k-1} \\ \bm{l} \not \ni i}} 1 
= \delta(i,j) \frac{1}{k} \binom{n-1}{k-1}.
\end{align*}
(Here we have identified the strictly increasing sequence $\bm{l}=(l_1<l_2< \cdots<l_{k-1})$ with 
the set $\{l_1,l_2,\dots, l_{k-1}\}$.)
Therefore, just as in the bosonic case, we find
\[
\mathbb{E} \left[\pi(H)_{i,j} \right]= \Tr(\lambda) \frac{\delta(i,j)}{n}.
\]

More generally, for the identity permutation $\mathrm{id}_p \in \mathfrak{S}_p$, 
we can see that
\begin{equation}
\Delta_{\mathrm{id}_p}^{n,k}(\bm{i},\bm{j})= \frac{1}{k^p}\delta(\bm{i},\bm{j}) 
\binom{n-1}{k-1}^p = \delta(\bm{i},\bm{j}) \left[\frac{1}{n} \binom{n}{k}\right]^p.
\end{equation}

\subsection{Some further observations}

\begin{lem} \label{lem:Delta_fermion_trans}
For the transposition $(1 \ 2) \in \mathfrak{S}_2$, the fermionic $\Delta$-function defined in (\ref{eqn:Delta-fer}) is given by
\begin{align*}
\Delta_{(1 \ 2)}^{n,k} \big( (i_1,i_2), (j_1,j_2) \big)
&= \delta(i_1, j_1) \, \delta(i_2, j_2) \frac{1}{k^2} \binom{n-2}{k-2} + 
\delta(i_1, j_2) \, \delta(i_2, j_1) \frac{1}{k^2}\binom{n-2}{k-1} \\
&= \begin{cases}
\frac{1}{k^2} \binom{n-1}{k-1} & \text{if $i_1=i_2=j_1=j_2$}, \\
\frac{1}{k^2} \binom{n-2}{k-1} & \text{if $i_1=j_2 \not= i_2=j_1$}, \\
\frac{1}{k^2}\binom{n-2}{k-2} & \text{if $i_1=j_1 \not= i_2=j_2$}, \\
0 & \text{otherwise}.
\end{cases} 
\end{align*}
\end{lem}

\begin{proof}
The definition of 
$\Delta_{\sigma}^{n,k}\big( \bm{i}, \bm{j}\big)$ implies that 
\begin{multline*}
\Delta_{(1 \ 2)}^{n,k}\big( (i_1,i_2), (j_1,j_2) \big)
=  \frac{1}{k^2}
\sum_{\substack{ \bm{l}^{(1)} \in \mathcal{A}_{k-1} \\ \bm{l}^{(1)} \not \ni i_1, j_1 }}
\sum_{\substack{ \bm{l}^{(2)} \in \mathcal{A}_{k-1} \\ \bm{l}^{(2)} \not \ni i_2, j_2 }}
\prod_{s=1}^2 \sgn(i_s, \bm{l}^{(s)}) \sgn(j_s, \bm{l}^{(s)})
 \\
\times 
\delta \big( \mathrm{sr}(i_1, \bm{l}^{(1)}),  \mathrm{sr}(j_2, \bm{l}^{(2)}) \big) \times
\delta\big( \mathrm{sr}(i_2, \bm{l}^{(2)}),  \mathrm{sr}(j_1, \bm{l}^{(1)}) \big).
\end{multline*}
Only the terms such that
\begin{equation} \label{eq:fermion_sequence_condition1}
\mathrm{sr}(i_1, \bm{l}^{(1)})= \mathrm{sr}(j_2, \bm{l}^{(2)}) \qquad \text{and} \qquad 
\mathrm{sr}(j_1, \bm{l}^{(1)})= \mathrm{sr}(i_2, \bm{l}^{(2)}) 
\end{equation}
contribute.

\begin{enumerate}
\item[(i)] First, suppose that $i_1=j_2$. Then, by \eqref{eq:fermion_sequence_condition1}, $\bm{l}^{(2)}$ must coincide with $\bm{l}^{(1)}$, and so $i_2=j_1$.
Hence, in this case, 
\[
\Delta_{(1 \ 2)}^{n,k}\big( (i_1,i_2), (j_1,j_2) \big)
=  \frac{1}{k^2}
\sum_{\substack{ \bm{l}^{(1)} \in \mathcal{A}_{k-1} \\ \bm{l}^{(1)} \not \ni i_1, j_1 }}1
= \begin{cases}
\frac{1}{k^2} \binom{n-1}{k-1} & \text{if $i_1=j_1$}, \\
\frac{1}{k^2} \binom{n-2}{k-1} & \text{if $i_1 \not=j_1$}.
\end{cases}
\]
\item[(ii)] Next, suppose that $i_1 \not=j_2$. 
Then,  by \eqref{eq:fermion_sequence_condition1},
the following equations all hold.
\[
i_2,j_2 \in \bm{l}^{(1)}, \quad \bm{l}^{(2)}= \bm{l}^{(1)} \cup \{i_1\} \setminus \{j_2\} = \bm{l}^{(1)} \cup \{j_1\} \setminus \{i_2\}.  
\]
This is only valid when $i_1=j_1$ and $i_2=j_2$.
Hence, since $i_1 \not=i_2 (=j_2)$, we have
\[
\Delta_{(1 \ 2)}^{n,k}\big( (i_1,i_2), (j_1,j_2) \big)
=  \frac{1}{k^2}
\sum_{\substack{ \bm{l}^{(1)} \in \mathcal{A}_{k-1} \\ \bm{l}^{(1)} \ni i_2, \ \bm{l}^{(1)} \not \ni i_1 }}1
= \frac{1}{k^2} \binom{n-2}{k-2}.
\]
\end{enumerate}
\end{proof}

Lemma \ref{lem:Delta_fermion_trans} indicates that 
$\Delta^{n,k}_{\sigma}(\bm{i},\bm{j})$
is \emph{not} of the form 
$\delta_\sigma(\bm{i},\bm{j}) \times X_\sigma^{n,k}$.
This fact, which holds as well for the bosonic $\Delta$-function defined in (\ref{eqn:Delta-bos}), poses an obstacle to further simplifying the formulae in Theorems \ref{thm:bos-moments} and \ref{thm:fer-moments}.  We suspect that is is computationally difficult to calculate the general moments in both cases.

\section{Asymptotic analysis} \label{sec:asymptotics}

We now return to the setting of two distinguishable particles, introduced in \S\ref{subsec:our_q}.  As an application of Theorem \ref{thm:projection_distribution}, we study the asymptotics of the moments in two regimes: the limit as $m \to \infty$ with $n$ fixed, and the limit as $n \to \infty$ with $m = \lceil cn \rceil$ for $c > 0$.

To pose questions about the asymptotics of the moments, we will need to consider \textit{sequences} of spectra $\lambda^{(k)} = (\lambda^{(k)}_{ij})_{ij=11}^{mn} \in \mathbb{R}^{mn}$, where $m = m(k)$ and $n = n(k)$ for $k = 1, 2, 3, \hdots$, and we will generally need to impose some kind of convergence condition on $\lambda^{(k)}$ as $k \to \infty$.  To this end, it is useful to introduce the \textit{empirical spectral measure}
\begin{equation} \label{eqn:esm-def}
\mu[\lambda^{(k)}] = \frac{1}{mn} \sum_{ij=11}^{mn} \delta_{\lambda^{(k)}_{ij}},
\end{equation}
where $\delta_{\lambda^{(k)}_{ij}}$ is a Dirac mass at $\lambda^{(k)}_{ij} \in \bR$.  We then can make the further assumption that $\mu[\lambda^{(k)}]$ converges weakly to some compactly supported probability measure $\mu$ on $\bR$, which will allow us to describe the asymptotic behavior of the quantum marginals in terms of the moments of this limiting measure $\mu$.  Recall that for $p \ge 1$, the $p$th \textit{moment} of a probability measure $\mu$ on $\bR$ is the quantity
\begin{equation}
M_p(\mu) = \int_{-\infty}^\infty x^p \, d\mu(x).
\end{equation}
If $(\mu_k)_{k \ge 1}$ is a sequence of probability measures supported in some bounded interval $[-K, K] \subset \bR$ and $\mu$ is another probability measure supported in $[-K, K]$, then $\mu_k$ converges weakly to $\mu$ if and only if $M_p(\mu_k) \to M_p(\mu)$ for all $p = 1, 2, 3, \hdots$ as $k \to \infty$.

Recall that for a permutation $\sigma \in \mathfrak{S}_d$, we write $\kappa(\sigma)$ for the total number of cycles in its cycle decomposition. Let $(c_1, \hdots, c_{\kappa(\sigma)})$ be the cycle type of $\sigma$.  We define
\begin{equation}
M_\sigma(\mu) = \prod_{j=1}^{\kappa(\sigma)} M_{c_j}(\mu).
\end{equation}
Our calculations below will make use of the following fact.
\begin{lem} \label{lem:trace-moments}
Let $\sigma \in \mathfrak{S}_d$.  For $\lambda =(\lambda_{ij})_{ij=11}^{mn} \in \mathbb{R}^{mn}$, define the empirical spectral measure $\mu[\lambda]$ as in (\ref{eqn:esm-def}).  Then
\begin{equation}
\Tr_\sigma(\lambda) = (mn)^{\kappa(\sigma)} M_\sigma(\mu[\lambda]).
\end{equation}
\end{lem}

\begin{proof}
Again write $(c_1, \hdots, c_{\kappa(\sigma)})$ for the cycle type of $\sigma$.  By definition,
$$\Tr_\sigma(\lambda) = \prod_{l=1}^{\kappa(\sigma)} \Tr(\diag(\lambda)^{c_l}),$$
so it suffices to show that $\Tr(\diag(\lambda)^{c_l}) = mn \cdot M_{c_l}(\mu[\lambda])$.  Indeed, we find
\[
	\Tr(\diag(\lambda)^{c_l}) = \sum_{ij=11}^{mn} \lambda_{ij}^{c_l} = mn \cdot \int_{-\infty}^\infty x^{c_l} \, d \bigg( \frac{1}{mn} \sum_{ij=11}^{mn} \delta_{\lambda_{ij}}(x) \bigg),
\]
as desired.
\end{proof}

\subsection{Large $m$, fixed $n$}
Here we take $n$ fixed and consider the limit as $m \to \infty.$
Recall that in the setting of \S\ref{subsec:our_q}, $H$ is a random Hermitian matrix given by
$H=(H_{ij, kl}) = U \diag (\lambda) U^\dagger$, where
$\lambda =(\lambda_{ij})_{ij=11}^{mn} \in \mathbb{R}^{mn}$ and $U$ is a Haar-distributed unitary matrix of size $mn$.
The marginals of $H$ are the $m \times m$ Hermitian matrix $\pi_1(H) =(\pi_1(H)_{i,k})$ and 
$n \times n$ Hermitian matrix $\pi_2(H)=(\pi_2(H)_{j,l})$ defined by
\begin{equation*} 
\pi_1(H)_{i,k}= \sum_{j=1}^n H_{ij, kj}, \qquad \pi_2(H)_{j,l}= \sum_{i=1}^m H_{ij, il}.
\end{equation*}
For sequences of indices
\begin{align*}
\bm{i}=(i_1,i_2,\dots,i_p), \ \bm{k}=(k_1,k_2,\dots,k_p)  \quad &\in \{1,2,\dots\}^{\times p}, \\ 
\bm{j}=(j_1, j_2, \dots, j_q), \ \bm{l}=(l_1,l_2,\dots,l_q)  \quad &\in \{1,2,\dots\}^{\times q},
\end{align*}
we define
\[
  s(\bm{i}, \bm{k}) = \max \big \{ \kappa(\sigma) \ \big | \ \sigma \in \mathfrak{S}_{p}, \ \sigma(\bm{i}) = \bm{k} \big \},
\]
\[
  s(\bm{j}, \bm{l}) = \max \big \{ \kappa(\tau) \ \big | \ \tau \in \mathfrak{S}_{q}, \ \tau(\bm{j}) = \bm{l} \big \}.
\]

\medskip

\begin{thm} \label{thm:asymp-m}
Let $p, q$ be nonnegative integers. Fix $n \ge 2$ and take sequences of indices
\begin{align*}
\bm{i}=(i_1,i_2,\dots,i_p), \ \bm{k}=(k_1,k_2,\dots,k_p)  \quad &\in \{1,2,\dots\}^{\times p}, \\ 
\bm{j}=(j_1, j_2, \dots, j_q), \ \bm{l}=(l_1,l_2,\dots,l_q)  \quad &\in \{1,2,\dots,n\}^{\times q}.
\end{align*}
Let $\mathfrak{m} = \max \{ i_1, \hdots, i_p, k_1, \hdots, k_p \}$ and choose a sequence $(\lambda^{(m)})_{m \ge \mathfrak{m}}$, where $\lambda^{(m)} \in \mathbb{R}^{mn}$.  Set $H^{(m)} = U(m) \diag(\lambda^{(m)}) U(m)^\dagger$, where $U(m)$ is an $mn \times mn$ Haar unitary.  Suppose that the measures $\mu[\lambda^{(m)}]$ all have support contained in some bounded interval $[-K, K] \subset \bR$, satisfy a uniform bound $| M_k(\mu[\lambda^{(m)}])| < C$ for all $k$ and $m$, and converge weakly to a probability measure $\mu$.  Then as $m \to \infty$,
\begin{multline} \label{eqn:asymp-m}
\mathbb{E} \left[ 
\prod_{\alpha=1}^p \pi_1(H^{(m)})_{i_\alpha, k_\alpha} \cdot 
\prod_{\beta=1}^q \pi_2(H^{(m)})_{j_\beta, l_\beta}
\right] \\
=  \left[ 
\sum_{\substack{\sigma \in \mathfrak{S}_{p} \\ \kappa(\sigma) = s(\bm{i}, \bm{k})}} \sum_{\substack{ \tau \in \mathfrak{S}_{q} \\ \kappa(\tau) = s(\bm{j}, \bm{l})}} \delta_\sigma (\bm{i},\bm{k}) \, \delta_\tau (\bm{j},\bm{l}) \, M_\sigma(\mu) M_\tau(\mu) \right] n^{2 s(\bm{i}, \bm{k}) + s(\bm{j}, \bm{l}) -p-q} \, m^{s(\bm{i}, \bm{k}) + 2 s(\bm{j}, \bm{l})-p-q} \cdot \big(1 + o(1) \big).
\end{multline}
\end{thm}

\begin{proof}
We use Theorem \ref{thm:projection_distribution}:
\begin{multline*}
\mathbb{E} \left[ 
\prod_{\alpha=1}^p \pi_1(H^{(m)})_{i_\alpha, k_\alpha} \cdot 
\prod_{\beta=1}^q \pi_2(H^{(m)})_{j_\beta, l_\beta}
\right] \\
= \sum_{\sigma ,\tau \in \mathfrak{S}_{p+q}} 
\delta_{\pr_1 (\sigma)} (\bm{i},\bm{k}) \, \delta_{\pr_2 (\sigma)} (\bm{j},\bm{l}) \,
n^{\kappa_2(\sigma)} m^{\kappa_1(\sigma)} \Wg_{mn}(\sigma^{-1}\tau) \Tr_\tau(\lambda^{(m)}).
\end{multline*}
We pick up the leading term in the limit $m \to \infty$. 
From Lemma \ref{lem:Wg-asymp}, we have 
\[
\Wg_{mn}(\sigma^{-1}\tau) = 
\begin{cases}
(mn)^{-p-q} (1+O(m^{-1})) & \text{if $\sigma=\tau$}, \\
O(m^{-p-q-1}) & \text{otherwise.}
\end{cases}
\]
Observe that for $\sigma$ such that $\pr_2(\sigma)(\bm{j}) = \bm{l}$ we have
\[
m^{\kappa_1(\sigma)} = 
\begin{cases}
m^{s(\bm{j}, \bm{l})} & \text{if $\sigma \in \mathfrak{S}_p \times \mathfrak{S}_q$ and $\kappa(\pr_2(\sigma))$ is maximal}, \\
O(m^{s(\bm{j}, \bm{l})-1}) & \text{otherwise},
\end{cases}
\]
where we regard $\mathfrak{S}_p \times \mathfrak{S}_q$ as a subgroup of $\mathfrak{S}_{p+q}$ in the usual way.  For $\sigma = (\sigma_1, \sigma_2) \in \mathfrak{S}_p \times \mathfrak{S}_q$, we have $\kappa_1(\sigma) = \kappa(\sigma_2),$ $\kappa_2(\sigma) = \kappa(\sigma_1)$, and $\Tr_\sigma(\la^{(m)}) = \Tr_{\sigma_1}(\la^{(m)})\Tr_{\sigma_2}(\la^{(m)})$.  Therefore we see that, as $m \to \infty$,
\begin{align} 
\begin{split} \label{eqn:asymp-m-deriv}
\mathbb{E} &\left[ 
\prod_{\alpha=1}^p \pi_1(H^{(m)})_{i_\alpha, k_\alpha} \cdot 
\prod_{\beta=1}^q \pi_2(H^{(m)})_{j_\beta, l_\beta}
\right] \\
&= \sum_{\sigma \in \mathfrak{S}_{p}} \sum_{\substack{ \tau \in \mathfrak{S}_{q} \\ \kappa(\tau) = s(\bm{j}, \bm{l})}} 
\delta_{\sigma} (\bm{i},\bm{k}) \, \delta_\tau (\bm{j},\bm{l}) \,
n^{\kappa_2((\sigma, \tau))} m^{s(\bm{j}, \bm{l})} (mn)^{-p-q} \Tr_{(\sigma, \tau)} (\lambda^{(m)}) \cdot \big(1 + O(m^{-1}) \big) \\
&= \left[ 
\sum_{\sigma \in \mathfrak{S}_{p}} \sum_{\substack{ \tau \in \mathfrak{S}_{q} \\ \kappa(\tau) = s(\bm{j}, \bm{l})}} \delta_\sigma (\bm{i},\bm{k}) \, \delta_\tau (\bm{j},\bm{l}) \, n^{\kappa(\sigma)-(p+q)} \, \Tr_\sigma(\la^{(m)}) \Tr_\tau(\lambda^{(m)}) \right] m^{s(\bm{j}, \bm{l})-(p+q)} \cdot \big(1 + O(m^{-1}) \big) \\
&= \left[ \sum_{\sigma \in \mathfrak{S}_{p}} \sum_{\substack{ \tau \in \mathfrak{S}_{q} \\ \kappa(\tau) = s(\bm{j}, \bm{l})}} \delta_\sigma (\bm{i},\bm{k}) \, \delta_\tau (\bm{j},\bm{l}) \, n^{\kappa(\sigma)-(p+q)} \, (mn)^{\kappa(\sigma)} M_\sigma(\mu[\la^{(m)}]) \, (mn)^{s(\bm{j}, \bm{l})} M_\tau(\mu[\lambda^{(m)}]) \right] \\
& \qquad \qquad \qquad \qquad \qquad \qquad \qquad \qquad \qquad \qquad \qquad \qquad \qquad \qquad \cdot m^{s(\bm{j}, \bm{l})-(p+q)} \cdot \big(1 + O(m^{-1}) \big),
\end{split}
\end{align}
where in the last equality we have used Lemma \ref{lem:trace-moments}. Again discarding lower-order terms in $m$, we find that the outer sum can be restricted only to $\sigma \in \mathfrak{S}_{p}$ with $\kappa(\sigma) = s(\bm{i}, \bm{k})$, giving
\begin{align*}
\mathbb{E} &\left[ 
\prod_{\alpha=1}^p \pi_1(H^{(m)})_{i_\alpha, k_\alpha} \cdot 
\prod_{\beta=1}^q \pi_2(H^{(m)})_{j_\beta, l_\beta}
\right] \\
&= \left[ \sum_{\substack{\sigma \in \mathfrak{S}_{p} \\ \kappa(\sigma) = s(\bm{i}, \bm{k})}} \sum_{\substack{ \tau \in \mathfrak{S}_{q} \\ \kappa(\tau) = s(\bm{j}, \bm{l})}} \delta_\sigma (\bm{i},\bm{k}) \, \delta_\tau (\bm{j},\bm{l}) \,  M_\sigma(\mu[\la^{(m)}]) \, M_\tau(\mu[\lambda^{(m)}]) \right] n^{2 s(\bm{i}, \bm{k}) + s(\bm{j}, \bm{l}) -(p+q)} \\
& \qquad \qquad \qquad \qquad \qquad \qquad \qquad \qquad \qquad \qquad \qquad \qquad \cdot m^{s(\bm{i}, \bm{k}) + 2s(\bm{j}, \bm{l})-(p+q)} \cdot \big(1 + O(m^{-1}) \big).
\end{align*}
The desired result then follows from the assumption that $\mu[\lambda^{(m)}]$ converges weakly to $\mu$, so that as $m \to \infty$, $M_k(\mu[\lambda^{(m)}]) \to M_k(\mu)$ for all $k \ge 1$.
\end{proof}

\begin{example}
If $p = 0$, the group $\mathfrak{S}_p$ is trivial, and (\ref{eqn:asymp-m}) gives
\begin{equation} \label{eqn:asymp-m-p0}
\mathbb{E} \left[
\prod_{\beta=1}^q \pi_2(H^{(m)})_{j_\beta, l_\beta}
\right] 
=  \left[ \sum_{\substack{ \tau \in \mathfrak{S}_{q} \\ \kappa(\tau) = s(\bm{j}, \bm{l})}} \delta_\tau (\bm{j},\bm{l}) \, M_\tau(\mu) \right] n^{s(\bm{j}, \bm{l})-q} \, m^{2s(\bm{j}, \bm{l})-q} \cdot \big(1 + o(1) \big).
\end{equation}
Similarly, for $q=0$, we have
\begin{equation} \label{eqn:asymp-m-q0}
\mathbb{E} \left[ 
\prod_{\alpha=1}^p \pi_1(H^{(m)})_{i_\alpha, k_\alpha}
\right] \\
=  \left[ 
\sum_{\substack{ \sigma \in \mathfrak{S}_{p} \\ \kappa(\sigma) = s(\bm{i}, \bm{k})}} \delta_\sigma (\bm{i},\bm{k}) \, M_\sigma(\mu) \right] n^{2s(\bm{i}, \bm{k})-p} \, m^{s(\bm{i}, \bm{k})-p} \cdot \big(1 + o(1) \big).
\end{equation}
\end{example}

\begin{cor} \label{cor:largest-growth-m}
The fastest possible asymptotic growth of the mixed moment in Theorem \ref{thm:asymp-m}, with $p$ and $q$ fixed, occurs when $\bm{i} = \bm{k}$ and $\bm{j} = \bm{l}$.  In this case, as $m \to \infty$, we have
\begin{equation} \label{eqn:asymp-m-jl}
\mathbb{E} \left[ 
\prod_{\alpha=1}^p \pi_1(H^{(m)})_{i_\alpha, k_\alpha} \cdot 
\prod_{\beta=1}^q \pi_2(H^{(m)})_{j_\beta, j_\beta}
\right]
=  M_1(\mu)^{p+q} \, n^{p} \, m^{q} \cdot \big(1 + o(1) \big).
\end{equation}
\end{cor}

\begin{proof}
When $\bm{i} = \bm{k}$ and $\bm{j} = \bm{l}$, we have $s(\bm{i},\bm{k}) = p$ and $s(\bm{j},\bm{l}) = q$, the largest values these quantities can take.  In this case the sums in (\ref{eqn:asymp-m}) include only the term with $\sigma = \mathrm{id}_p$ and $\tau = \mathrm{id}_q$, with $M_\sigma(\mu) = M_1(\mu)^p$ and $M_\tau(\mu) = M_1(\mu)^q$, so that (\ref{eqn:asymp-m}) simplifies to (\ref{eqn:asymp-m-jl}).
\end{proof}

The above results imply the following concentration theorem, analogous to a law of large numbers for the scaled marginal $\pi_2(m^{-1} H^{(m)})$.

\begin{thm} \label{thm:fd-marg-limit}
In the setting of Theorem \ref{thm:asymp-m}, the random matrix $\pi_2(m^{-1} H^{(m)})$ converges in probability to the deterministic value $M_1(\mu) I_n$, where $I_n$ is the $n \times n$ identity matrix.
\end{thm}

\begin{proof}
From (\ref{eqn:asymp-m-p0}) and Corollary \ref{cor:largest-growth-m}, for $\bm{j}=(j_1, j_2, \dots, j_q) \in \{1,2,\dots,n\}^{\times q}$ we find
\[
\mathbb{E} \left[
\prod_{\beta=1}^q \pi_2(m^{-1}H^{(m)})_{j_\beta, j_\beta}
\right] 
=  M_1(\mu)^q + o(1)
\]
as $m \to \infty$, where we have used the linearity of both $\pi_2$ and $\mathbb{E}$.  All other mixed moments of entries of $\pi_2(m^{-1}H^{(m)})$ are $o(1)$.  These statements imply the limits
\begin{align*}
\mathbb{E}[\pi_2(m^{-1} H^{(m)})_{j, l}] &\to M_1(\mu) \, \delta(j,l), \\
\mathrm{Var}[\pi_2(m^{-1} H^{(m)})_{j, l}] &\to 0
\end{align*}
for $j,l = 1, \hdots, n$.  The theorem then follows from Chebyshev's inequality.
\end{proof}

\begin{remark} \label{rem:asymp-m-phys}
The assumptions of Theorem \ref{thm:asymp-m} on convergence of moments of the empirical measures are natural from the point of view of random matrix theory.  However, density matrices of physical quantum mechanical systems are constrained to be positive semidefinite and have trace 1.  Accordingly, in physical applications, it is more realistic instead to make the much stronger assumptions
\begin{align}
	\nonumber \Tr(\diag(\lambda^{(m)})) &= 1, \\
	\label{eqn:asymp-m-phys-assump} \Tr(\diag(\lambda^{(m)})^k) &\to T_k \ge 0 \ \textrm{ as $m \to \infty$,}
\end{align}
with a uniform bound $|\Tr(\diag(\lambda^{(m)})^k)| < C$ for all $m$ and $k$.  It is easy to see that these assumptions imply that the emprical measures $\mu[\lambda^{(m)}]$ converge in distribution to a Dirac mass at 0.  However, the results above are easily adjusted to this setting.  In particular, in place of (\ref{eqn:asymp-m}), we find
\begin{multline} \label{eqn:asymp-m-phys}
\mathbb{E} \left[ 
\prod_{\alpha=1}^p \pi_1(H^{(m)})_{i_\alpha, k_\alpha} \cdot 
\prod_{\beta=1}^q \pi_2(H^{(m)})_{j_\beta, l_\beta}
\right] \\
=  \left[ 
\sum_{\sigma \in \mathfrak{S}_{p}} \sum_{\substack{ \tau \in \mathfrak{S}_{q} \\ \kappa(\tau) = s(\bm{j}, \bm{l})}} \delta_\sigma (\bm{i},\bm{k}) \, \delta_\tau (\bm{j},\bm{l}) \, n^{\kappa(\sigma)-p-q} \, T_\sigma T_\tau \right] m^{s(\bm{j}, \bm{l})-p-q} \cdot \big(1 + o(1) \big),
\end{multline}
where for $\sigma \in \mathfrak{S}_d$ with cycle type $(c_1, \hdots, c_{\kappa(\sigma)})$, we define $T_\sigma = \prod_{j=1}^{\kappa(\sigma)} T_{c_j}$. The asymptotic formula (\ref{eqn:asymp-m-phys}) follows directly from the third line of (\ref{eqn:asymp-m-deriv}).  Note that here we \emph{cannot} restrict the sum over $\mathfrak{S}_{p}$ to include only $\sigma$ with $\kappa(\sigma) = s(\bm{i},\bm{k})$.

Additionally, in this setting all moments of the entries of the marginals are $O(1)$ as $m \to \infty$.  For fixed $p$, the slowest-decaying moments are those for which $\bm{j} = \bm{l}$, irrespective of whether $\bm{i} = \bm{k}$.  In place of (\ref{eqn:asymp-m-jl}) we find
\begin{multline} \label{eqn:asymp-m-jl-phys}
\mathbb{E} \left[ 
\prod_{\alpha=1}^p \pi_1(H^{(m)})_{i_\alpha, k_\alpha} \cdot 
\prod_{\beta=1}^q \pi_2(H^{(m)})_{j_\beta, j_\beta}
\right] \\
=  \left[ 
\sum_{\sigma \in \mathfrak{S}_{p}} \delta_\sigma (\bm{i},\bm{k}) \, n^{\kappa(\sigma)-p-q} \, T_\sigma \right] m^{-p} \cdot \big(1 + o(1) \big),
\end{multline}
which implies that in this case the marginal $\pi_2(H^{(m)})$ itself, rather than its rescaling $\pi_2(m^{-1}H^{(m)})$, converges in probability to the deterministic value $n^{-1} I_n$.
\end{remark}

\subsection{Large $m$ and $n$}

Here we study the limit as $m$ and $n$ both grow large with their ratio fixed.  That is, we send $n \to \infty$ with $m = \lceil cn \rceil$ for $c > 0$.

\begin{thm} \label{thm:asymp-n}
Let $p, q$ be nonnegative integers and fix some $c > 0$. Take sequences of indices
\begin{align*}
\bm{i}=(i_1,i_2,\dots,i_p), \ \bm{k}=(k_1,k_2,\dots,k_p)  \quad \in \{1,2,\dots\}^{\times p} &, \\ 
\bm{j}=(j_1, j_2, \dots, j_q), \ \bm{l}=(l_1,l_2,\dots,l_q)  \quad \in \{1,2,\dots\}^{\times q}. &
\end{align*}
Set $m = m(n) = \lceil cn \rceil$ for $n = 1, 2, \hdots$, and take $\mathfrak{n}$ sufficiently large that
\[
     j_1, \hdots, j_q, l_1, \hdots, l_q \le \mathfrak{n} \qquad \textrm{and} \qquad i_1, \hdots, i_p, k_1, \hdots, k_p \le m(\mathfrak{n}).
\]
Fix a sequence $(\lambda^{(n)})_{n \ge \mathfrak{n}}$, where $\lambda^{(n)} \in \mathbb{R}^{mn}$.  Set $H^{(n)} = U(n) \diag(\lambda^{(n)}) U(n)^\dagger$, where $U(n)$ is an $mn \times mn$ Haar unitary. Suppose that the measures $\mu[\lambda^{(n)}]$ all have support contained in some bounded interval $[-K, K] \subset \bR$, satisfy a uniform bound $|M_k(\mu[\lambda^{(n)}])| < C$ for all $k$ and $n$, and converge weakly to a probability measure $\mu$. Then as $n \to \infty$, 
\begin{multline} \label{eqn:asymp-n}
\mathbb{E} \left[ 
\prod_{\alpha=1}^p \pi_1(H^{(n)})_{i_\alpha, k_\alpha} \cdot 
\prod_{\beta=1}^q \pi_2(H^{(n)})_{j_\beta, l_\beta}
\right] \\
=  \left[ \sum_{\substack{ \sigma \in \mathfrak{S}_{p} \\ \kappa(\sigma) = s(\bm{i}, \bm{k})}} \sum_{\substack{ \tau \in \mathfrak{S}_{q} \\ \kappa(\tau) = s(\bm{j}, \bm{l})}} \delta_\sigma (\bm{i},\bm{k}) \, \delta_\tau (\bm{j},\bm{l}) \, M_\sigma(\mu) M_\tau(\mu) \right] c^{s(\bm{i}, \bm{k}) + 2s(\bm{j}, \bm{l})-p-q} \\
\cdot  n^{3(s(\bm{i}, \bm{k})+ s(\bm{j}, \bm{l})) - 2(p + q)} \cdot \big(1 + o(1) \big).
\end{multline}
\end{thm}

\begin{proof}
The proof is similar to Theorem \ref{thm:asymp-m}.  We again use Theorem \ref{thm:projection_distribution}, along with the fact that in this limiting regime, by Lemma \ref{lem:Wg-asymp},
\[
\Wg_{mn}(\sigma^{-1}\tau) = 
\begin{cases}
(cn^2)^{-p-q} (1+O(n^{-1})) & \text{if $\sigma=\tau$}, \\
O(n^{-2(p+q+1)}) & \text{otherwise},
\end{cases}
\]
to obtain
\begin{multline} \label{eqn:asymp-n-deriv}
\mathbb{E} \left[ 
\prod_{\alpha=1}^p \pi_1(H^{(n)})_{i_\alpha, k_\alpha} \cdot 
\prod_{\beta=1}^q \pi_2(H^{(n)})_{j_\beta, l_\beta}
\right] \\
= \sum_{\sigma \in \mathfrak{S}_{p+q}} 
\delta_{\pr_1 (\sigma)} (\bm{i},\bm{k}) \delta_{\pr_2 (\sigma)} (\bm{j},\bm{l})
n^{\kappa_2(\sigma)} (cn)^{\kappa_1(\sigma)} (cn^2)^{-p-q} \Tr_\sigma(\lambda^{(n)}) \cdot \big( 1+O(n^{-1}) \big) \\
= \sum_{\sigma \in \mathfrak{S}_{p+q}} 
\delta_{\pr_1 (\sigma)} (\bm{i},\bm{k}) \delta_{\pr_2 (\sigma)} (\bm{j},\bm{l}) \, c^{\kappa(\sigma) + \kappa_1(\sigma) -p-q} \, n^{2 \kappa(\sigma) + \kappa_2(\sigma) + \kappa_1(\sigma) -2(p+q)} M_\sigma(\mu[\lambda^{(n)}]) \cdot \big( 1+O(n^{-1}) \big),
\end{multline}
where in the last equality we have used Lemma \ref{lem:trace-moments}.  The leading-order contribution comes from $\sigma \in \mathfrak{S}_{p+q}$ such that $\kappa_1(\sigma) + \kappa_2(\sigma)$ is maximized among $\sigma$ satisfying $\pr_1 (\sigma) (\bm{i}) = \bm{k}$, $\pr_2 (\sigma) (\bm{j}) = \bm{l}$.  This can occur only when $\sigma = (\pr_1(\sigma), \pr_2(\sigma)) \in \mathfrak{S}_p \times \mathfrak{S}_q$ with $\kappa(\pr_1(\sigma)) = s(\bm{i}, \bm{k})$, $\kappa(\pr_2(\sigma)) = s(\bm{j}, \bm{l})$.  Extracting only such terms from the sum and using the convergence $M_k(\mu[\lambda^{(n)}]) \to M_k(\mu)$ for all $k \ge 1$ as $n \to \infty$, we obtain the desired result.
\end{proof}

\begin{cor}
The fastest possible asymptotic growth of the mixed moment in Theorem \ref{thm:asymp-n}, with $p$ and $q$ fixed, occurs when $\bm{i} = \bm{k}$ and $\bm{j} = \bm{l}$.  In this case, as $n \to \infty$, we have
\begin{equation} \label{eqn:asymp-n-ik-jl}
\mathbb{E} \left[ 
\prod_{\alpha=1}^p \pi_1(H^{(n)})_{i_\alpha, i_\alpha} \cdot 
\prod_{\beta=1}^q \pi_2(H^{(n)})_{j_\beta, j_\beta}
\right] = M_1(\mu)^{p+q} \, c^q \, n^{p+q} \cdot \big( 1 + o(1) \big).
\end{equation}
\end{cor}

\begin{proof}
For any $\bm{i}, \bm{j}, \bm{k}, \bm{l}$, we have $s(\bm{i}, \bm{k}) \le p$ and $s(\bm{j}, \bm{l}) \le q$.  These bounds are saturated only when $\bm{i} = \bm{k}$ and $\bm{j} = \bm{l}$, in which case the only leading-order term in (\ref{eqn:asymp-n}) is the term with $\sigma = \mathrm{id}_p \in \mathfrak{S}_p$ and $\tau = \mathrm{id}_q \in \mathfrak{S}_q$.
\end{proof}

\begin{remark} \label{rem:asymp-n-phys}
In this regime as well, we can make the alternative, stronger convergence assumptions, analogous to (\ref{eqn:asymp-m-phys-assump}):
\begin{align}
	\nonumber \Tr(\diag(\lambda^{(n)})) &= 1, \\
	\label{eqn:asymp-n-phys-assump} \Tr(\diag(\lambda^{(n)})^k) &\to T_k \ge 0 \ \textrm{ as $n \to \infty$,}
\end{align}
with a uniform bound $|\Tr(\diag(\lambda^{(n)})^k)| < C$ for all $n$ and $k$.  Again, these assumptions imply that the emprical measures $\mu[\lambda^{(n)}]$ converge in distribution to a Dirac mass at 0, but the results above can be rescaled to apply in this setting.  In place of (\ref{eqn:asymp-n}) we find
\begin{multline} \label{eqn:asymp-n-phys}
\mathbb{E} \left[ 
\prod_{\alpha=1}^p \pi_1(H^{(n)})_{i_\alpha, k_\alpha} \cdot 
\prod_{\beta=1}^q \pi_2(H^{(n)})_{j_\beta, l_\beta}
\right] \\
=  \left[ \sum_{\substack{ \sigma \in \mathfrak{S}_{p} \\ \kappa(\sigma) = s(\bm{i}, \bm{k})}} \sum_{\substack{ \tau \in \mathfrak{S}_{q} \\ \kappa(\tau) = s(\bm{j}, \bm{l})}} \delta_\sigma (\bm{i},\bm{k}) \, \delta_\tau (\bm{j},\bm{l}) \, c^{s(\bm{j}, \bm{l})-p-q} T_\sigma T_\tau \right] n^{s(\bm{i}, \bm{k})+ s(\bm{j}, \bm{l}) - 2(p + q)} \, \cdot \big(1 + o(1) \big),
\end{multline}
as can be obtained directly from (\ref{eqn:asymp-n-deriv}), while in place of (\ref{eqn:asymp-n-ik-jl}) we find
\begin{equation} \label{eqn:asymp-n-ik-jl-phys}
\mathbb{E} \left[ 
\prod_{\alpha=1}^p \pi_1(H^{(n)})_{i_\alpha, i_\alpha} \cdot 
\prod_{\beta=1}^q \pi_2(H^{(n)})_{j_\beta, j_\beta}
\right] = c^{-p} n^{-p-q} \cdot \big( 1 + o(1) \big).
\end{equation}
\end{remark}

\section*{Acknowledgments}
The work of the first author (SM) was supported by JSPS KAKENHI Grant Numbers 17K05281, 22K03233.  The work of the second author (CM) was supported by the National Science Foundation under grant number DMS 2103170, and by JST CREST program JPMJCR18T6.  The authors would like to thank Beno\^it Collins for helpful discussions during the preparation of this manuscript.


\bigskip


\begin{thebibliography}{AAA}

\bibitem{Barysh} Baryshnikov, Y., ``{GUE}s and queues,'' {\it Probab. Theory Related Fields} \textbf{119} (2001), 256--274

\bibitem{BGH} S. Belinschi, A. Guionnet, J. Huang, ``Large deviation principles via spherical integrals'' {\it Probab. Math. Phys.} \textbf{3} (2022), 543--625, \url{arXiv:2004.07117}

\bibitem{BHKRV} E. Bianchi, L. Hackl, M. Kieburg, M. Rigol, L. Vidmar, ``Volume-law entanglement entropy of typical pure quantum states'' {\it PRX Quantum} \textbf{3} (2022), 030201, \url{arXiv:2112.06959}

\bibitem{BI} P. B\"urgisser, C. Ikenmeyer, ``The complexity of computing Kronecker coefficients,'' in \textit{Discrete Math. Theoret. Comput. Sci. Proc.} vol. AJ, \textit{20th Annual International Conference on Formal Power Series and Algebraic Combinatorics (FPSAC 2008)}, Nancy (2008), 357--368

\bibitem{CDKW} M. Christandl, B. Moran, S. Kousidis, M. Walter, ``Eigenvalue distributions of reduced density matrices, {\it Comm. Math. Phys.} \textbf{322} (2014), 1--52, \url{arXiv:1204.0741}

\bibitem{CY} A. J. Coleman, V. I. Yukalov, \textit{Lecture Notes in Chemistry} vol. 72, \textit{Reduced Density Matrices: Coulson's Challenge}, New York: Springer (2000)

\bibitem{Collins03} B. Collins, ``Moments and cumulants of polynomial random variables on unitary groups, the Itzykson--Zuber integral, and free probability,'' {\it Int. Math. Res. Not. IMRN} \textbf{2003} (2003), 953--982, \url{arXiv:math-ph/0205010}

\bibitem{CMN-review} B. Collins, S. Matsumoto, J. Novak, ``The Weingarten calculus,'' {\it Notices Amer. Math. Soc.} \textbf{69} (2022), 734--745, \url{arXiv:2109.14890}

\bibitem{CMS}
B. Collins, S. Matsumoto, N. Saad, 
``Integration of invariant matrices and moments of inverses of Ginibre and Wishart matrices,''
\textit{J. Multivariate Anal.} \textbf{126} (2014), 1--13

\bibitem{CM} B. Collins, C. McSwiggen, ``Projections of orbital measures and quantum marginal problems,'' to appear in {\it Trans. Amer. Math. Soc.} (2023), \url{arXiv:2112.13908}

\bibitem{CMZ} R. Coquereaux, C. McSwiggen, J.-B. Zuber, {``On Horn's problem and its volume function},'' {\it Comm. Math. Phys.} \textbf{376} (2020), 2409--2439, \url{arXiv:1904.00752}

\bibitem{CMZ2} R. Coquereaux, C. McSwiggen, J.-B. Zuber, {``Revisiting Horn's problem''}, {\it J. Stat. Mech. Theory Exp} (2019), 094018, \url{arXiv:1905.09662}

\bibitem{CZ1} R. Coquereaux, J.-B. Zuber, {``From orbital measures to Littlewood--Richardson coefficients and hive polytopes''}, {\it Ann. Inst. Henri Poincar\'e D} {\bf 5} (2018), 339--386,
\url{arXiv:1706.02793}

\bibitem{CZ2} R. Coquereaux, J.-B. Zuber, ``The Horn problem for real symmetric and quaternionic self-dual matrices,'' {\it SIGMA Symmetry Integrability Geom. Methods Appl.} \textbf{15} (2019), 029, \url{arXiv:1809.03394}

\bibitem{CZ-SchurKostka} R. Coquereaux, J.-B. Zuber, {``On Schur problem and Kostka numbers''}, in {\it Proceedings of Symposia in Pure Mathematics} vol. 103.2, {\it Integrability, Quantization, and Geometry: II. Quantum Theories and Algebraic Geometry}, I. Krichever, S. Novikov, O. Ogievetsky, S. Shlosman eds., Providence: American Mathematical Society (2021), 111--136, \url{arXiv:2001.08046}

\bibitem{CC-corners} C. Cuenca, ``Universal behavior of the corners of orbital beta processes,'' {\it Int. Math. Res. Not. IMRN} (2019), rnz226, \url{arXiv:1807.06134}

\bibitem{ForresterZhang} P.J. Forrester, J. Zhang, ``Co-rank 1 projections and the randomised Horn problem,'' {\it Tunis J. Math.} \textbf{3} (2021), 55--73

\bibitem{Ful97} W. Fulton, \textit{Young Tableaux: With Applications to Representation Theory and Geometry}, London: London Mathematical Society (1997)

\bibitem{Hausdorff-moments} F. Hausdorff, ``Momentprobleme f\"ur ein endliches Intervall'' (German), \textit{Math. Z.} \textbf{16} (1923), 220--248.

\bibitem{IMW} C. Ikenmeyer, K. D. Mulmuley, M. Walter, ``On vanishing of Kronecker coefficients,'' \textit{Comput. Complexity} \textbf{26} (2017), 949--992, \url{arXiv:1507.02955}

\bibitem{IP} C. Ikenmeyer, G. Panova, ``Rectangular Kronecker coefficients and plethysms in Geometric Complexity Theory,'' in \textit{2016 IEEE 57th Annual Symposium on Foundations of Computer Science (FOCS)}, New Brunswick (2016), 396--405

\bibitem{Knut-lectures} A. Knutson, ``Schubert calculus and quantum information,'' paper presented at {\it Workshop on Quantum Marginals and Density Matrices}, Fields Institute, Toronto (2009), \url{http://pi.math.cornell.edu/~allenk/qclectures.pdf}

\bibitem{Klyachko-marginals} A. Klyachko, ``Quantum marginal problem and representations of the symmetric group'' (2004), \url{arXiv:quant-ph/0409113}

\bibitem{Liu} Y.-K. Liu, ``Consistency of local density matrices is QMA-complete,'' in {\it Proceedings of the 10th International Workshop on Randomization and Computation, RANDOM 2006}, Barcelona (2006), 438--449

\bibitem{LCV} Y.-K. Liu, M. Christandl, F. Verstraete, ``Quantum computational complexity of the $N$-representability problem: QMA complete,'' \textit{Phys. Rev. Lett.} \textbf{98} (2007), 110503

\bibitem{LloydPagels} S. Lloyd, H. Pagels, ``Complexity as thermodynamic depth,'' {\it Ann. Physics} \textbf{188} (1988), 186--213

\bibitem{McS-splines} C. McSwiggen, ``Box splines, tensor product multiplicities and the volume function,'' {\it Algebr. Comb.} \textbf{4} (2021), 435--464, \url{arXiv:1909.12278}

\bibitem{Stillinger} National Research Council, \textit{Mathematical Challenges from Theoretical/Computational Chemistry}, Washington, DC: National Academies Press (1995)

\bibitem{Ruskai} M. B. Ruskai, ``$N$-representability problem: conditions on geminals,'' \textit{Phys. Rev.} \textbf{183} (1969), 129--141

\bibitem{Schilling-review} C. Schilling, ``The quantum marginal problem,'' in {\it QMATH 12 -- Mathematical Results in Quantum Mechanics}, Humboldt University, Berlin (2013), \url{arXiv:1404.1085}

\bibitem{SZ1} H.-J. Sommers, K. Zyczkowski, ``Statistical properties of random density matrices,'' {\it J. Phys. A} \textbf{37} (2004), 8457

\bibitem{SZ2} H.-J. Sommers, K. Zyczkowski, ``Bures volume of the set of mixed quantum states,'' {\it J. Phys. A} \textbf{36} (2003), 10083

\bibitem{Tyc-Vlach} T. Tyc, J. Vlach, ``Quantum marginal problems,'' {\it Eur. Phys. J. D} \textbf{69} (2015)

\bibitem{ZKF} J. Zhang, M. Kieburg, P.J. Forrester, ``Harmonic analysis for rank-1 randomised Horn problems,'' {\it Lett. Math. Phys.} \textbf{111} (2021), 98, \url{arXiv:1911.11316}

\bibitem{Z} J.-B. Zuber, ``Horn's problem and Harish-Chandra's integrals. Probability density functions,'' {\it Ann. Inst. Henri Poincar\'e D} {\bf 5} (2018), 309--338, \url{arXiv:1705.01186}

\bibitem{Z2} J.-B. Zuber, ``On the minor problem and branching coefficients,'' {\it Ann. Inst. Henri Poincar\'e D} {\bf 9} (2022), 349--366, \url{arXiv:2006.03006}

\bibitem{SZ3} K. Zyczkowski, H.-J. Sommers, ``Induced measures in the space of mixed quantum states,'' {\it J. Phys. A} \textbf{34} (2001), 7111

\end{thebibliography}
\end{document}